\pdfoutput=1
\RequirePackage{ifpdf}
\ifpdf % We are running pdfTeX in pdf mode
\documentclass[pdftex]{sigma}
\else
\documentclass{sigma}
\fi

\def\a{\alpha}
\def\b{\beta}
\def\d{\delta}

\def\g{\gamma}

\numberwithin{equation}{section}

\newtheorem{Theorem}{Theorem}[section]
\newtheorem{Corollary}[Theorem]{Corollary}
\newtheorem{Proposition}[Theorem]{Proposition}

{ \theoremstyle{definition}
\newtheorem{Remark}[Theorem]{Remark} }

\begin{document}
\allowdisplaybreaks

\newcommand{\arXivNumber}{2005.?????}

\renewcommand{\PaperNumber}{047}

\FirstPageHeading

\ShortArticleName{New Separation of Variables for the Classical XXX and XXZ Heisenberg Spin Chains}

\ArticleName{New Separation of Variables for the Classical XXX\\ and XXZ Heisenberg Spin Chains}

\Author{Guido MAGNANO~$^\dag$ and Taras SKRYPNYK~$^\dag{}^\ddag$}

\AuthorNameForHeading{G.~Magnano and T.~Skrypnyk}

\Address{$^\dag$~Universit\`a degli Studi di Torino, via Carlo Alberto 10, 10123, Torino, Italia}
\EmailD{\href{mailto:guido.magnano@unito.it}{guido.magnano@unito.it}, \href{mailto:taras.skrypnyk@unito.it}{taras.skrypnyk@unito.it}}

\Address{$^\ddag$~Bogolyubov Institute for Theoretical Physics, Kyiv, Ukraine}

\ArticleDates{Received December 20, 2019, in final form May 16, 2020; Published online June 02, 2020}

\Abstract{We propose a non-standard separation of variables for the classical integrable XXX and XXZ spin chains with degenerate twist matrix. We show that for the case of such twist matrices one can interchange the role of classical separating functions $A(u)$ and $B(u)$ and construct a new full set of separated variables, satisfying simpler equation of separation and simpler Abel equations in comparison with the standard separated variables of Sklyanin. We show that for certain cases of the twist matrices the constructed separated variables can be directly identified with action-angle coordinates.}


\Classification{37J35; 17B80}

\section{Introduction}

Completely integrable Hamiltonian systems admitting Lax representation \cite{Lax} have been an object of constant interest in physics and mathematics during the last forty years.

An important problem in the theory of integrable systems still to be solved in general is the problem of variable separation. The separated variables $x_i$, $p_j$, $i,j=
 1,\dots,N $ are a set of (quasi)canonical coordinates such that the following system of {equations} is
satisfied~\cite{SklSep}
\begin{gather*}%\label{sepeq0}
 \Phi_i(x_i, p_i,I_1, \dots ,I_N, C_1,\dots,C_r) = 0, \qquad i=1,\dots,N,
 \end{gather*}
 where $\Phi_i$ are certain functions, $I_k$ are Poisson-commuting integrals of motion, $C_i$ are Casimir functions and $N$ is half of the dimension of the phase space.

The separated coordinates (whenever they exist) allow one to write Abel-type equations (see Section~\ref{section2.1}) which, in their turn, provide a possibility to solve explicitly the Hamilton equations of motion upon resolving the corresponding Abel--Jacobi inversion problem.
 Separated variables are also important when solving quantum
integrable models~\cite{SklSep}. That is why the construction of variable separation is
 a central issue in the theory of both classical and quantum integrable systems.

 In order to construct separated variables in many cases one can use a pair of \emph{separating functions} $A(u)$, $B(u)$, which depend on the dynamical variables and on a complex parameter $u$. The coordinates $x_1, \dots, x_N$ are determined as zeroes of $B(u)$,
\begin{gather}\label{xi}
B(x_i)=0, \qquad i=1,\dots,N,
\end{gather}
and (quasi)momenta are obtained as values of $A(u)$ in these zeros
\begin{gather}\label{pi}
p_i=A(x_i),\qquad i=1,\dots,N .
\end{gather}

In the so-called Lax-integrable case (i.e., when the Hamilton equations of motion can be written in Lax form), a prescription to obtain separating functions $B(u)$, $A(u)$ for $\mathfrak{gl}(2)$-valued Lax matrices $L(u)=\sum\limits_{i,j=1}^{2}L_{ij}(u)X_{ij} $, where $X_{ij}$, $i,j=1,2$ is a standard basis of
$\mathfrak{gl}(2)$ $(X_{ij})_{\a\b}=\d_{i\a}\d_{j,\b}$, have been introduced in \cite{SklTrig, SklRat,SklSep}\footnote{Often in the literature the choice $B(u)=L_{12}(u)$ is made. Nevertheless, since for the Cartan-invariant $r$-matrices the choice between $L_{12}(u)$ or $L_{21}(u)$ is arbitrary, we prefer the above-made choice.}
\begin{gather}\label{sfstand}
B(u)=L_{21}(u), \qquad A(u)=L_{11}(u).
\end{gather}
The equations of separation in this case coincide with the spectral curve of the Lax matrix
\begin{gather*}
 \Phi_i(x_i, p_i, I_1, \dots ,I_N, C_1,\dots,C_r) = \det (L(x_i)-p_i \, {\rm Id})=0, \qquad i=1,\dots,N.
\end{gather*}
The constructed variables $x_i$, $p_i$ should be (quasi)canonical in order for the theory to work
\begin{gather*}
\{ x_i, p_j\}=\d_{ij} f_i(x_i,p_i), \qquad \{x_i,x_j\}=0,\qquad \{p_i,p_j\}=0, \qquad \forall\,
i,j= 1,\dots,N.
\end{gather*}
The separating functions $A(u)$ and $B(u)$ have to satisfy an appropriate Poisson algebra \cite{SkrDub,SkrDub2} in order to produce (quasi)canonical coordinates (see Section~\ref{section2.2} below).
 The right-hand side of this algebra, as well as the very definition (\ref{xi})--(\ref{pi}), is asymmetric in the functions~$A(u)$,~$B(u)$. Nevertheless, in some cases the algebra of separating functions is symmetric in~$A(u)$ and~$B(u)$ and has the following form
\begin{subequations}\label{sepalg0}
\begin{gather} \label{sepalg01}
 \{B(u),B(v)\}=b(u,v)A(v)B(u)- b(v,u)A(u)B(v),
\\ \label{sepalg02}
 \{A(u),B(v)\}=\a(u,v)A(v)B(u)- \b(u,v)A(u)B(v),
\\ \label{sepalg03}
 \{A(u),A(v)\}=a(u,v)A(v)B(u)- a(v,u)A(u)B(v).
\end{gather}
\end{subequations}
 This situation occurs when the Lax matrix satisfy quadratic tensor brackets \cite{TF, Skl}
 \begin{gather}\label{rmbr0}
\{{L}(u)\otimes 1,1\otimes
{L}(v)\}=[r(u,v),{L}(u)\otimes L(v)].
\end{gather}
Here $r(u,v)=\sum\limits_{i,j,k,l=1}^{2}r_{ij,kl}(u, v)X_{ij} \otimes
X_{kl}$, $X_{ij}$, is a skew-symmetric classical $r$-matrix.

The symmetry of the separating algebra (\ref{sepalg0}) poses a natural question: is it possible (when the separating algebra is symmetric) to exchange the roles of the separating functions? That is, is it possible to define new separated variables as follows
\begin{gather}\label{xi'pi'}
A(x_i)=0, \qquad p_i=B(x_i),\qquad i=1,\dots,N ?
\end{gather}

The answer to this question is not obvious. Indeed, while the reversed definition (\ref{xi'pi'}) in the symmetric case guarantees the (quasi)canonicity of the constructed separated coordinates, it does not guarantee the existence of the equations of separation for them.\footnote{See Remark~\ref{remark5} for another -- ``banal'' explanation of the reversed definition (\ref{xi'pi'}) which is not used here.}

In the present paper we are going to answer two general questions:
\begin{enumerate}\itemsep=0pt
\item For what $\mathfrak{gl}(2)\otimes \mathfrak{gl}(2)$-valued $r$-matrices the separating algebra of functions $A(u)$ and $B(u)$ defined by (\ref{sfstand}) is symmetric in $A(u)$ and $B(v)$, i.e., has the form (\ref{sepalg0})?
\item When does the reversed definition (\ref{xi'pi'}) produce separated variables for such $r$-matrices? That is, when do the corresponding quasi-canonical coordinates satisfy equations of separation with the initial algebra of first integrals?
\end{enumerate}

For the convenience of the reader we formulate our answers already in the Introduction. In particular, the answer to the first question is contained in the following proposition.
\begin{Proposition} The functions $A(u)$ and $B(u)$ defined by~\eqref{sfstand} satisfy the algebra \eqref{sepalg0} with respect to the brackets~\eqref{rmbr0} if the components of the $r$-matrix satisfy the following conditions
\begin{subequations}\label{condrms0}
\begin{gather}\label{condrm01}
r_{21,21}(u,v)=0, \qquad r_{21,11}(u,v)=0, \qquad r_{11,21}(u,v)=0,
\\
\label{condrm02}
r_{12,12}(u,v)=0, \qquad r_{12,22}(u,v)=0, \qquad r_{22,12}(u,v)=0,
\\
\label{condrm03}
 r_{22,22}(u,v)=r_{11,11}(u,v).
\end{gather}
\end{subequations}
\end{Proposition}

 There are at least two $\mathfrak{gl}(2)\otimes \mathfrak{gl}(2)$ valued classical skew-symmetric $r$-matrices satisfying the condition (\ref{condrms0}): the
 standard rational and standard trigonometric $r$-matrices. For the case of these $r$-matrices we proceed with the answer to the second question. For this purpose we define also the class of the Lax matrices under the consideration. We consider the most physically important case of the Lax matrices of the XXX and XXZ models of $N$ classical spins with a~twisted periodic boundary conditions, i.e., the Lax matrices of the following form
 \begin{gather}\label{laxc}
L(u)=L^{(\nu_1)}(u) \cdots L^{(\nu_N)}(u) C,
\end{gather}
where $C=\sum\limits_{i,j=1}^2 c_{ij}X_{ij}$ is a two by two twist matrix satisfying the following condition
\begin{gather}\label{rCC}
[r(u,v),C\otimes C]=0
\end{gather}
and $L^{(\nu_i)}(u)$ is the Lax matrix with a simple pole in the point $u=\nu_i$, corresponding to $i$-th cite of the classical spin chain, where the classical spins satisfy the Poisson brackets of $\mathfrak{gl}(2)^{\oplus N}$ (in the case of the rational $r$-matrix) and of the direct sum of $N$ trigonometric Sklyanin-type algebras (in the case of trigonometric $r$-matrix). Observe that for the case of all the considered quadratic Poisson algebras the non-trivial integrals of motion are generated by $I(u)=\operatorname{tr} L(u)$.

The following theorem holds true:

\begin{Theorem}%\label{mainteo0}
\looseness=-1 Let the coordinates $x_i$, $p_i$, $i=1,\dots,N $ be defined by \eqref{xi'pi'}, the functions $A(u)$ and $B(u)$ be defined by \eqref{sfstand} and the Lax matrix $L(u)$ be defined by \eqref{laxc}. Let the mat\-rix~$C$ be such that $c_{11}\neq 0$. Then the coordinates~$x_i$, $p_i$ are separated coordinates for the classical integrable system with the algebra of integrals of motion generated by $I(u)=\operatorname{tr} L(u)$ if and only if
\begin{gather*} \det C=0.\end{gather*}
\end{Theorem}

For the rational XXX case the matrix $C=\sum\limits_{i,j=1}^2 c_{ij}X_{ij}$ is an arbitrary constant degenerated matrix. In this case the case the equations of separation are written as follows
\begin{gather} \label{eqsep0}
c_{12} p_i=c_{11} I(x_i), \qquad i=1,\dots,N,
\end{gather}
and we distinguish to cases: $c_{12}\neq 0$ and $c_{12}=0$. In the first case the curve of separation is a~rational one and the corresponding Abel equations are the following
\begin{gather*}
\sum\limits_{i=1}^N \dfrac{c_{11}x_i^{N-j}}{c_{12}p_i}
\frac{{\rm d}x_i}{{\rm d}t_k} =-\d^j_{k}, \qquad j,k=1,\dots,N ,
\end{gather*}
where $\frac{{\rm d}x_i}{{\rm d}t_k}=\{I_k,x_i\}$, $k=1,\dots,N $ etc.\ are easily integrated in terms of elementary functions. This is a consequence of the fact that the separation curve~(\ref{eqsep0}) is a rational one.

The second case $c_{12}=0$ is even more special. In this case we also have that $c_{22}=0$ and the equations of separation coincides with separating polynomial and acquires the form
\begin{gather}\label{specsepeq2}
 I(x_i)=A(x_i)=0, \qquad i=1,\dots,N .
\end{gather}
The coordinates $x_i$ in this case become functions of the integrals of motion and may be identified with action variables. The Abel equations are written for the (quasi)conjugate variables $p_i$
\begin{gather*}
\sum\limits_{i=1}^N \dfrac{x_i^{N-j}}{\partial_{x_i} I(x_i) }\frac{1}{p_i}
\frac{{\rm d}p_i}{{\rm d}t_k} =\d^j_{k}, \qquad j,k=1,\dots,N ,
\end{gather*}
and produce linear differential equations for the corresponding angle coordinates $\phi_i=\ln p_i$.

In the trigonometric XXZ case the constant twist matrix $C$ is diagonal, \mbox{$C=\sum\limits_{i=1}^2 c_{ii}X_{ii}$} and $c_{22}=0$ in the degenerated case. That is why in this case there is only a special degenerated case of non-standard separation of variables characterized by the equation of separation~(\ref{specsepeq2}), where the coordinates of separation coincide with action variables and the canonically conjugated variables coincide with the angles of the Liouville theorem.

To the best of our knowledge these are the first examples (at least in the Lax-integrable case) when the variables of separation coincide with the action-angle variables and their construction provides immediate solution of the equations of motion without performing of a (generally speaking difficult) task of solving of the Abel--Jacobi inversion problem.

At the end of the introduction let us make several bibliographical comments. The Lax-pair based approach to the variable separation in its general form was proposed by Sklyanin in~\cite{SklSep} as a development of his previous idea \cite{SklToda,SklTrig, SklRat}. In the classical case the idea of the approach may be traced further back to the papers \cite{Alber,NV}. In the quantum case the approach has obtained a lot of attention in the literature, let us mention only the series of recent papers \cite{Maillet17,MailletJPA18, MailletJMP18,RV}.
 In the classical case, unfortunately, there have been very few works on the subject. For the classical XXX model we can mention several papers \cite{DD,Geht,Scott,SklRat,SklYan,SklSep}.
For the classical XXZ model we can mention only two papers on the subject \cite{SklTrig} and \cite{SkrDub2}. To fill this gap in the knowledge and to study the corresponding classical models in more details is one of the aims of our paper.

The structure of the present paper is the following: in Section~\ref{section2} we remind general notions of the classical variable separation theory, in Section~\ref{section3} we consider Lax-integrable case, in Sections~\ref{section4} and \ref{section5} we concentrate on the examples of the classical XXX and XXZ models. In these sections we also consider $N=2$ examples, investigating the corresponding cases in details. In particular, we explicitly find the reconstruction formulae for them, expressing the initial dynamical variables via the constructed coordinates of separation and the values of the Casimir functions. At last, in Section~\ref{section5} we conclude and discuss the open problems.

\section{Separation of variables}\label{section2}
\subsection{Definitions and notations}\label{section2.1}
Let us recall the definitions of Liouville integrability and
separation of variables in the general theory of Hamiltonian
systems. An integrable Hamiltonian system with $N$ degrees of
freedom is determined on a $2N$-dimensional symplectic manifold
$\mathcal{M}$ (symplectic leaf in $(\mathcal{P},\{\ ,\ \})$ and~$N$ independent functions (first integrals) $I_j$ commuting with respect to the Poisson bracket{\samepage
\begin{gather*}\{I_i ,I_j\} = 0 ,\qquad i,j =1,\dots,N \end{gather*}
(for the Hamiltonian $H$ of the system may be taken any first integral $I_j$).}

 To find separated variables means to find
(at least locally) a set of coordinates $x_i$, $p_j$, $i,j =1,\dots,N$ such that there exist $N$ relations
\begin{gather*}
 \Phi_i(x_i, p_i, I_1,\dots, I_N, C_1,\dots,C_r) = 0, \qquad i=1,\dots,N,
 \end{gather*}
where $C_i$, $i= 1,\dots,r $ are Casimir functions and the coordinates $x_i$, $p_j$, $i,j =1,\dots,N $ are
canonical
\begin{gather*}
\{ x_i, p_j\}=\d_{ij}, \qquad \{x_i,x_j\}=0,\qquad \{p_i,p_j\}=0, \qquad \forall\, i,j =1,\dots,N .
\end{gather*}
In the present paper it will be convenient for us to work with quasi-canonical coordinates satisfying the following Poisson brackets
\begin{gather*}%\label{qcan}
\{ x_i, p_j\}=\d_{ij} p_j, \qquad \{x_i,x_j\}=0,\qquad \{p_i,p_j\}=0, \qquad \forall\, i,j =1,\dots,N .
\end{gather*}
Clearly the variables $x_i$, $\phi_j=\log p_j$ will be canonical then.

It is possible to show that the coordinates of separation $x_i$ satisfy the Abel-type equations
\begin{gather}\label{abel0}
\sum\limits_{i=1}^N \dfrac{ \dfrac{\partial \Phi_i(x_i, p_i, I_1,\dots, I_N, C_1,\dots,C_r)}{\partial I_k} }{ \dfrac{\partial \Phi_i(x_i, p_i ;I_1,\dots, I_N, C_1,\dots,C_r)}{\partial p_i}} \frac{1}{ p_i}\frac{{\rm d}x_i}{{\rm d}t_j} = - \d_{kj}, \qquad \forall\, k,j=1,\dots,N,
\end{gather}
and similar Abel-type equations are satisfied by the momenta of separation $p_i$
\begin{gather}\label{abelMom0}
\sum\limits_{i=1}^N \dfrac{ \dfrac{\partial \Phi_i(x_i, p_i, I_1,\dots, I_N, C_1,\dots,C_r)}{\partial I_k} }{ \dfrac{\partial \Phi_i(x_i, p_i ;I_1,\dots, I_N, C_1,\dots,C_r)}{\partial x_i}} \frac{1}{ p_i}\frac{{\rm d}p_i}{{\rm d}t_j} = \d_{kj}, \qquad \forall\, k,j=1,\dots,N .
\end{gather}
These equations are the last step before the integration of the classical equations of motion.

\subsection{The method of separating functions}\label{section2.2}
Let $B(u)$ and $A(u)$ be some functions of the dynamical variables and of an auxiliary complex parameter $u$, which is constant with respect to the bracket $\{\ ,\ \}$. Let the points $x_i$, $i=1,\dots,N $ be zeros of the function $B(u)$ and $p_i$, $i=1,\dots,N $ be the values of $A(u)$ in these points. We wish to obtain (quasi)canonical Poisson brackets among these new coordinates using the Poisson brackets between functions $B(u)$ and $A(u)$. The following proposition holds true \cite{SkrDub}.
\begin{Proposition}\label{canon}
Let the coordinates $x_i$ and $p_j$, $i,j=1,\dots,p$ be defined as $B(x_i)=0$, $p_j=A(x_j)$. Let the functions $A(u)$, $B(u)$ satisfy the following Poisson algebra
\begin{subequations}\label{sepalg}
\begin{gather} \label{sepalg1}
 \{B(u),B(v)\}=b(u,v)B(u)- b(v,u)B(v),\\
 \label{sepalg2} \{A(u),B(v)\}=\a(u,v)B(u)- \b(u,v)B(v),\\
 \label{sepalg3} \{A(u),A(v)\}=a(u,v)B(u)- a(v,u)B(v).
\end{gather}
\end{subequations} Then the Poisson bracket {between} the functions $x_i$ and $p_j$, $\forall\, i,j=1,\dots,N $ are the following
\begin{gather*}
 \{x_i,x_j\}=0, \qquad \forall\, i,j=1,\dots,N,\\
 \{x_j,p_i\}=0, \qquad \text{if} \quad i\neq j,\\
 \{p_i,p_j\}=0, \qquad \forall\, i,j=1,\dots,N .
\end{gather*}
 If, moreover also the condition
\begin{gather*} %\label{sepalg4}
 \lim\limits_{u\rightarrow v} (\a(u,v)B(u)- \b(u,v)B(v))= A(v) \partial_v B(v) +\g(v) B(v)
\end{gather*}
holds, then the corresponding Poisson brackets are quasi-canonical, i.e.,
\begin{gather*}
 \{ x_i, p_i\}=p_i, \qquad \forall\, i=1,\dots,N .
\end{gather*}
\end{Proposition}

\begin{Remark}\label{remark1}Observe that the coefficients $a(u,v)$, $b(u,v)$, $\a(u,v)$, $\b(u,v)$, $\g(v)$ above may depend not only on the spectral parameters but also on the dynamical variables.
\end{Remark}

\begin{Remark}\label{remark2} Observe that in general separating algebra (\ref{sepalg}) is asymmetric in the functions $A(u)$, $B(v)$. Nevertheless for some dynamical coefficients $a(u,v)$, $b(u,v)$,
$\a(u,v)$, $\b(u,v)$ it may become symmetric in the functions $A(u)$, $B(v)$. In such a case the functions $A(u)$, $B(v)$ become interchangeable and one can ``invert'' the procedure, defining separated coordinates also as follows: $A(x_i)=0$, $p_j=B(x_j)$. This is the situation that will be studied in this article.
\end{Remark}

\section{Separation of variables: Lax-integrable case}\label{section3}
\subsection{The equations of separation}\label{section3.1}
Let us specify the above theory, i.e., equations of separation and separating functions for the Lax-integrable case, when Hamiltonian equations of motion with respect to a Hamiltonian $H$ can be written in Lax form \cite{Lax} with a spectral-parameter-dependent Lax matrix
\begin{gather*}%\label{algint}
\dot L(u)=\left[L(u), M_H(u)\right]
\end{gather*}
According to the ``magic recipe'' of Sklyanin in this case the role of all equations of separation is played by a single equation, namely the spectral curve of the Lax matrix
 \begin{gather*}%\label{curva}
 \Phi_i(x_i, p_i, I_1,\dots, I_N, C_1,\dots,C_r) = \det (L(x_i) - p_i\, {\rm Id})=0, \qquad i=1,\dots,N .
 \end{gather*}
This hypothesis works good for the case of the $\mathfrak{gl}(n)$-valued Lax matrices \cite{AHH, DD, Geht, Scott,SklSep}. In what follows we will consider the simplest case of the $\mathfrak{gl}(2)$-valued Lax matrices.
\subsection{The separating functions}\label{section3.2}
Let $X_{ij}$, $i,j=1,2$ be a standard basis in $\mathfrak{gl}(2)$ with the
commutation relations
\begin{gather*}%\label{comrel1}
[ X_{ij}, X_{kl}]= \d_{kj} X_{il}-\d_{il}X_{kj}.
\end{gather*}
The $\mathfrak{gl}(2)$-valued Lax matrix is written as follows
\begin{gather*}L(u)=\sum\limits_{i,j=1}^2 L_{ij}(u)X_{ij}.\end{gather*}
Following the ``magic recipe'' in its standard version \cite{SklSep} we will assume that the separating functions $A(u)$ and $B(u)$ are defined as follows:
\begin{gather}\label{AB}
A(u)=L_{11}(u), \qquad B(u)=L_{21}(u).
\end{gather}

\subsection{The separating algebra and its symmetries}\label{section3.3}
Now we will require that the algebra of the functions $A(u)$ and $B(u)$ defined by~(\ref{AB}) have the particular form~(\ref{sepalg}). For this purpose it is necessary at first to define the Poisson brackets among the components of the Lax matrix.

In this paper we will consider the case of the so-called quadratic Sklyanin bracket~\cite{Skl}
\begin{gather}\label{rmbr}
\{{L}(u)\otimes 1,1\otimes {L}(v)\}=\big[r^{12}(u,v),{L}(u)\otimes L(v)\big],
\end{gather}
where
\begin{gather}\label{rmcf}
r^{12}(u,v)=\sum\limits_{i,j,k,l=1}^{2}r_{ij,kl}(u, v)X_{ij}
\otimes X_{kl}
\end{gather}
is a skew-symmetric classical $r$-matrix: $r^{12}(u,v)=-r^{21}(v,u)$ (see \cite{BD, TF, RF, Skl}).

The algebra (\ref{sepalg}) is satisfied by the above functions $A(u)$ and $B(u)$ under certain conditions on the $r$-matrix. In more detail, the following proposition holds true.\footnote{See \cite{SkrDub,SkrDub2} for the generalization of this proposition onto the case of $\mathfrak{gl}(n)$-valued Lax matrices.}
\begin{Proposition} The functions $A(u)$ and $B(u)$ defined by \eqref{AB} satisfy the algebra \eqref{sepalg} with respect to the brackets \eqref{rmbr} if the components of the $r$-matrix satisfy the following conditions
\begin{gather}\label{condrmg}
r_{21,21}(u,v)=0, \qquad r_{21,11}(u,v)=0, \qquad r_{11,21}(u,v)=0.
\end{gather}
\end{Proposition}

\begin{proof} The proof of the proposition is achieved by direct calculation.\end{proof}

Observe that the algebra (\ref{sepalg}) is very asymmetric in the functions $A(u)$ and $B(u)$. It is asymmetric also in the considered case after imposing the conditions~(\ref{condrmg}).
We wish to investigate the question when the quadratic Poisson algebra of the functions $A(u)$ and $B(u)$ not only satisfies~(\ref{sepalg}) but also has its right-hand side to be symmetric in functions~$A(u)$ and~$B(u)$. Evidently, such symmetry will require more rigid conditions on the $r$-matrix.

The following proposition holds true.
\begin{Proposition}\quad
\begin{enumerate}\itemsep=0pt
\item[$(i)$] The functions $A(u)$ and $B(u)$ defined by \eqref{AB} satisfy the algebra \eqref{sepalg} with respect to the brackets \eqref{rmbr} in symmetric way with respect to $A(u)$ and $B(v)$ if the components of the $r$-matrix satisfy the following conditions
\begin{subequations}\label{condrms}
\begin{gather}\label{condrm1}
r_{21,21}(u,v)=0, \qquad r_{21,11}(u,v)=0, \qquad r_{11,21}(u,v)=0,
\\ \label{condrm2}
 r_{12,12}(u,v)=0, \qquad r_{12,22}(u,v)=0, \qquad r_{22,12}(u,v)=0,
\\ \label{condrm3}
 r_{22,22}(u,v)=r_{11,11}(u,v).
\end{gather}
\end{subequations}
\item[$(ii)$] Under the conditions \eqref{condrms} the algebra of separating functions reads as follows
\begin{subequations}\label{sepalgs}
\begin{gather} \label{sepalgs1}
 \{B(u),B(v)\}=r_{21,22}(u,v) A(u) B(v)+r_{22,21}(u,v) B(u) A(v),
\\ \label{sepalgs2}
 \{A(u),B(v)\}= (r_{11,22}(u,v)-r_{11,11}(u,v))B(v) A(u)+r_{12,21}(u,v) B(u) A(v),
\\ \label{sepalgs3}
 \{A(u),A(v)\}= r_{11,12}(u,v) A(u)B(v)+r_{12,11}(u,v) B(u)A(v).
\end{gather}
\end{subequations}
\end{enumerate}
\end{Proposition}

\begin{proof} The proposition is proven by direct calculation.\end{proof}

\begin{Remark}\label{remark3}There are at least two $\mathfrak{gl}(2)\otimes \mathfrak{gl}(2)$ valued classical skew-symmetric $r$-matrices satisfying the condition (\ref{condrms}):
 standard rational and standard trigonometric $r$-matrices.
\end{Remark}

Using the symmetry of separating algebra, one can invert the procedure described in the previous subsection and define the (quasi)canonical coordinates as follows
 \begin{gather}\label{p-qnonst}
A(x_i)=0, \qquad p_i=B(x_i).
\end{gather}
In the next section we will address, for the cases of rational and trigonometric $r$-matrices, the question whether the canonical coordinates defined by (\ref{p-qnonst}) are separation coordinates.

\section{Classical XXX spin model}\label{section4}
\subsection{Poisson brackets and Lax matrix}\label{section4.1}
Let us now consider the simplest possible case of the standard rational $r$-matrix
\begin{gather}\label{rrm}
r(u,v)=\frac{\sum\limits_{i,j=1}^2 X_{ij}\otimes X_{ji}}{u-v}
\end{gather}
and describe the corresponding Lax matrices $L(u)$, satisfying the quadratic brac\-kets~(\ref{rmbr}).

As it is well-known, the Lax matrices of the spin-chain models satisfying the quadratic brackets~(\ref{rmbr}) can be written in the following product form
\begin{gather}\label{prodLraz}
L^{(1,2,\dots,N)}(u)= L^{(\nu_1)}(u)\cdots L^{(\nu_N)}(u),
\end{gather}
where $N$ is arbitrary and the basic one-spin matrices $L^{(i)}(u)$ are written as follows
\begin{gather*}%\label{onepoleLraz}
L^{(\nu_k)}(u)=1+\frac{1}{u-\nu_k} \sum\limits_{i,j=1}^2 S^{(k)}_{ji} X_{ij}.
\end{gather*}
Here $\nu_i\neq \nu_j$, when $i\neq j$, $i,j=1,\dots,N $ and the Poisson brackets among the coordinates $S^{(n)}_{ij}$, $S^{(m)}_{kl}$ are those of the direct sum $\mathfrak{gl}(2)^{\oplus N}$
\begin{gather*}%\label{comrel}
\big\{{S}_{ij}^{(m)}, {S}_{kl}^{(n)}\big\}= \d^{mn}\big(\d_{kj}
{S}_{il}^{(m)}- \d_{il} {S}_{kj}^{(m)}\big).
\end{gather*}

Hereafter we will be interested in the Lax matrices of the following form
\begin{gather}\label{LCraz}
L(u)= L^{(1,2,\dots,N)}(u) C,
\end{gather}
where $C$ is an arbitrary constant matrix $C=\sum\limits_{i,j=1}^2 c_{ij} X_{ij}$.
By the virtue of the fact that in the case of the $r$-matrix (\ref{rrm}) \begin{gather*}[r(u,v), C\otimes C]=0,\end{gather*} the Lax matrix (\ref{LCraz}) satisfies the Poisson bracket (\ref{rmbr}) with a rational $r$-matrix (\ref{rrm}).

The Lax matrix (\ref{LCraz}) is a Lax matrix of the {\it XXX type $\mathfrak{gl}(2)$ Heisenberg spin chain with twisted periodic boundary conditions} defined by the constant twist matrix $C$ \cite{SklSep}.

\subsection{Integrals and Casimir functions}\label{section4.2}
Let us now consider the integrable system on $\mathfrak{gl}(2)^{\oplus N}$ defined with the help of the Lax mat\-rix~(\ref{LCraz}). Its mutually Poisson-commuting integrals of motion are constructed from the characteristic polynomial of the Lax matrix $L(u)$
\begin{gather*}I(w,u)=\det (L(u)-w \, {\rm Id})=w^2-w \operatorname{tr} L(u)+ \det L(u).\end{gather*}
The function $ \det L(u)=\det C \det L^{(1,2,\dots,N)}(u)$ is a generating function of the Casimir functions of the quadratic brackets~(\ref{rmbr}). The corresponding Casimir functions it contains consist of different combinations of the linear and quadratic Casimir functions of $\mathfrak{gl}(2)^{\oplus N}$
\begin{gather*}
c_{k}=S^{(k)}_{11}+ S^{(k)}_{22}, \qquad C_{k}=S^{(k)}_{11} S^{(k)}_{22}- S^{(k)}_{12} S^{(k)}_{21}.
\end{gather*}
The generating function of integrals of motion is
\begin{gather*}
I(u)= \operatorname{tr} L(u).
\end{gather*}
More explicitly, we have that
\begin{gather*}
I(u)=\frac{1}{\prod\limits_{i=1}^N (\nu_i-u)}\left( (-1)^N(c_{11}+c_{22}) u^{N}+ \sum\limits_{k=1}^{N} I_k u^{N-k}\right),
\end{gather*}
where the Hamiltonians $I_k$ are non-homogeneous polynomials of the degree up to $k$ in the dynamical variables. In particular, we have that
\begin{gather*}
I_1=-\left( \sum\limits_{i,j=1}^2 c_{ij} \sum\limits_{k=1}^N S^{(k)}_{ij}+ \sum\limits_{l=1}^N \nu_l \sum\limits_{i=1}^2 c_{ii}\right),\\
I_2= \sum\limits_{i,j=1}^2 c_{ij} \sum\limits_{l,m=1,\, l<m}^N \sum\limits_{k=1}^2 S^{(m)}_{ik}S^{(l)}_{kj}+ \sum\limits_{i,j=1}^2 c_{ij} \sum\limits_{k=1}^N \left(\sum\limits_{l=1}^N \nu_l-\nu_k\right) S^{(k)}_{ij} + \sum\limits_{k,l=1,\, k<l}^N \nu_l \nu_k \sum\limits_{i=1}^2 c_{ii}.
\end{gather*}

The integral $I_2$ may be viewed as a quadratic Hamiltonian of the classical spin chain of $N$ spins. It evidently does not have the Heisenberg-type form. The Heisenberg-type Hamiltonians are obtained from $I(u)$ (in the quantum case it works in the representation of $\mathfrak{gl}(2)^{\oplus N}$ algebra by Pauli matrices) using formulae of the following type~\cite{SklSep}
\begin{gather*}
H_n=\left(\frac{{\rm d}^n \ln I(u)}{{\rm d} u^n}\right)\bigg|_{u=0}.
\end{gather*}
We will not consider or use the Hamiltonians $H_n$ here: they do not enter directly into the spectral curve nor into the equation of separation, and therefore are not useful for our purposes.

\begin{Remark} \label{remark4} We have preserved seemingly non-important constant terms in the integrals of motion because they make non-trivial contributions to the equations of separation.
\end{Remark}

\subsection{Non-standard variable separation: general degenerated case}\label{section4.3}
\subsubsection{Separated variables and equation of separation}\label{section4.3.1}
The rational $r$-matrix (\ref{rrm}) evidently satisfies the conditions (\ref{condrms}). Hence, using our inverted definition (\ref{p-qnonst}) one can define the (quasi)canonical coordinates as follows
\begin{gather*}
A(x_i)=0, \qquad p_i=B(x_i),
\end{gather*}
where the functions $A(u)$ and $B(u)$ are standardly defined using the Lax matrix $L(u)$
\begin{gather*}
A(u)=L_{11}(u),\qquad B(u)=L_{21}(u).
\end{gather*}
For the Lax operator defined by the formulae (\ref{prodLraz}), (\ref{LCraz}) we obtain
\begin{gather*}
A(u)= \frac{1}{\prod\limits_{i=1}^N (\nu_i-u)}\left( (-1)^N c_{11} u^{N}+ \sum\limits_{k=1}^{N} A_k u^{N-k}\right),
\\
B(u)= \frac{1}{\prod\limits_{i=1}^N (\nu_i-u)}\left( (-1)^N c_{21} u^{N}+ \sum\limits_{k=1}^{N} B_k u^{N-k}\right),
\end{gather*}
where $A_k$, $B_k$ are the functions of the dynamical variables $S^{(l)}_{ij}$ and the constants $c_{ij}$.
From their explicit form it follows that in the case $c_{11}\neq 0$ one can take $A(u)$ for the role of separating polynomial: it produces~$N$ separated coordinates~$x_i$.

The following theorem holds true.
\begin{Theorem}\label{mainteo} Let us assume that the matrix $C$ is such that $c_{11}\neq 0$.\footnote{In the case $c_{11}=0$, $c_{22}\neq 0$ one should redefine the separating functions: $A(u)=L_{22}(u)$, $B(u)=L_{12}(u).$}
Let the coordinates $x_i$, $p_i$, $i=1,\dots,N $ be defined by~\eqref{p-qnonst}, where $A(u)=L_{11}(u)$, $B(u)=L_{21}(u)$ and the Lax matrix $L(u)$ be defined by the formulae \eqref{prodLraz}, \eqref{LCraz}. Then the coordinates $x_i$, $p_i$ are the separated coordinates for the integrable system with the algebra of integrals of motion generated by $I(u)=\operatorname{tr} L(u)$ if and only if \begin{gather*}\det C\equiv c_{11} c_{22}-c_{12} c_{21}=0.\end{gather*}
The equations of separation have in this case the following form
\begin{gather}\label{eqsepnst1}
c_{12} p_i=c_{11} I(x_i), \qquad i=1,\dots,N .
\end{gather}
\end{Theorem}

\begin{Remark}\label{remark5}
One may wonder if the exchange of roles between $A(u)$ and $B(u)$ is not trivially equivalent to swapping the Lax matrix elements $L_{11}(u)$ and $L_{21}(u)$ by a mere change of Lax representation. However, this is not so: indeed, one might take as a new Lax matrix the product $L^{C_0}\equiv C_0 L(u)$, where $C_0=(X_{12}+X_{21})$.\footnote{In the quantum case a similar trick was used in the paper~\cite{RV}.} In this case we have \begin{gather*}B^{C_0}(u)\equiv L^{C_0}_{21}(u)=L_{11}(u)=A(u), \qquad A^{C_0}(u)\equiv L^{C_0}_{11}(u)=L_{21}(u)=B(u).\end{gather*}
But in this way the integrals of motion are also changed --- instead of the generating function
\begin{gather*}I(u)= \operatorname{tr}L(u)= L_{11}(u)+L_{22}(u)\end{gather*}
 one obtains a new generating function
 \begin{gather*}I^{C_0}(u)= \operatorname{tr} L^{C_0}(u)= L_{12}(u)+L_{21}(u).\end{gather*}
After the change, the Sklyanin recipe yields separation variables for \emph{another} integrable system with a different algebra of first integrals, and the equations of separation coincide with the spectral curve of the Lax matrix $L^{C_0}(u)$: no progress is obtained for the original Lax equation. We stress that this is totally different from the method that we are introducing here, where the inversion between $A(u)$ and $B(u)$ produces a different set of separated variables for the same integrable system.
\end{Remark}

\begin{proof}[Proof of Theorem~\ref{mainteo}]
 First of all let us observe, that by the virtue of the relation (\ref{sepalgs}) and the explicit form of the $r$-matrix we obtain
\begin{gather*}% \label{sepalgs1}
 \{B(u),B(v)\}=0,
\\ % \label{sepalgs2}
 \{A(u),B(v)\}= \frac{B(u) A(v)- B(v) A(u) }{u-v},
\\ % \label{sepalgs3}
 \{A(u),A(v)\}= 0.
\end{gather*}
From this and the Proposition~\ref{canon} (with the reversed role of the function $A(u)$ and $B(u)$) it immediately follows that the coordinates $x_i$, $p_i$ defined as above are (quasi)canonical
\begin{gather*}
 \{x_i, x_j\}=0, \qquad \{p_i, p_j\}=0 , \qquad \{ x_i, p_i\}=p_i, \qquad \forall\, i=1,\dots,N .
\end{gather*}

Now it is left to show that they satisfy the equations of separation. We will prove them for any Lax matrix of the form
\begin{gather*}L(u)=\tilde{L}(u) C,\end{gather*}
where $\tilde{L}(u)$ is any Lax matrix satisfying quadratic brackets (\ref{rmbr}) and $C$ is a degenerated constant matrix satisfying~(\ref{rCC}).
Taking into account the explicit form of the functions~$B(u)$ and~$A(u)$ in terms of the components of the Lax matrix $\tilde{L}(u)$ and the matrix $C$ we obtain
\begin{gather*}
B(u) = \tilde{L}_{21}(u) c_{11}+\tilde{L}_{22}(u) c_{21}, \qquad A(u)= \tilde{L}_{11}(u) c_{11}+\tilde{L}_{12}(u) c_{21},
\\
I(u)= \tilde{L}_{11}(u) c_{11}+ \tilde{L}_{12}(u) c_{21} + \tilde{L}_{21}(u) c_{12}+ \tilde{L}_{22}(u) c_{22}.
\end{gather*}
From this it immediately follows that the following identity holds true
\begin{gather}\label{Iu}
I(u)=\frac{c_{12}}{c_{11}}B(u)+A(u)+ \frac{c_{11} c_{22}-c_{12} c_{21}}{c_{11}} \tilde{L}_{22}(u).
\end{gather}
Using this we obtain that in the case when $c_{11} c_{22}-c_{12} c_{21}=0$ the (quasi)canonical coordinates~$x_i$,~$p_i$ satisfy the following equations of separation
\begin{gather*}
I(x_i)=\frac{c_{12}}{c_{11}}p_i, \qquad i=1,\dots,N.
\end{gather*}
This proves the sufficiency of the conditions $c_{11} c_{22}-c_{12} c_{21}=0$ for separability with our ``inversed'' separating functions $A(u)$ and $B(u)$.

The necessity follows from the results of \cite{Magri} that provide necessary and sufficient differential (with respect to the time flows) conditions for the coefficients of the separating polynomial ($A(u)$~in our case) to generate separated coordinates. Making use of these conditions and calculating them explicitly for our system we obtain that they are satisfied if and only if $c_{11} c_{22}-c_{12} c_{21}=0$.
\end{proof}

\begin{Remark}\label{remark6} Observe that in the case $c_{21}\neq 0$ one can use more standard scheme taking the function $B(u)$ to be separating polynomial. In the case $c_{21}=0$ the choice of $B(u)$ to be separating polynomial is a bad choice, due to the fact that $B(u)$ in this case has less roots than necessary. If we also have that $c_{22}=0$, $c_{11}\neq 0$, then by the virtue of the Theorem above the good choice of the separating polynomial is the function $A(u)$.
\end{Remark}

\begin{Remark}\label{remark7} Observe that the equation of separation $I(x_i)=\frac{c_{12}}{c_{11}} p_i$ can be also obtained from the spectral curve of the Lax matrix $L(u)$. Indeed, putting $\det (C)=0$ we obtain that the spectral curve degenerates to the form
\begin{gather*}
w(w-I(u))=0.
\end{gather*}
If we now put \begin{gather*}u=x_i, \qquad w=\dfrac{c_{12}}{c_{11}} p_i\end{gather*} we obtain the equation~(\ref{eqsepnst1}).
\end{Remark}

\begin{Remark}\label{remark8} Observe also that if we had chosen the coordinates as follows: $p_i=A(x_i)$, $B(x_i)=0$ under the condition $\det (C)=0$ the equation of separation would have been $I(x_i)= p_i$, not $c_{12} p_i=c_{11} I(x_i)$ which makes a difference, especially in the case $c_{12}=0$.
\end{Remark}

\subsubsection{The Abel equations}\label{section4.3.2}
The important step in the theory of classical variable separation is re-writing of the equations of motion in the Abel form. This permits one to integrate the equations of motion for $x_i$ resolving the so-called Abel--Jacobi inversion problem.

The following proposition holds true.

\begin{Proposition}The equations of motion for the coordinates of separation $x_i$ described in the previous subsection are written in the Abel-type form as follows
\begin{gather}\label{Abel}
\sum\limits_{i=1}^N \dfrac{x_i^{N-j}}{(-1)^N(c_{11}+c_{22}) x_i^{N}+ \sum\limits_{k=1}^{N} I_k x_i^{N-k}}
\frac{{\rm d}x_i}{{\rm d}t_k} =-\d^j_{k}, \qquad j,k=1,\dots,N ,
\end{gather}
where $\frac{{\rm d}x_i}{{\rm d}t_k}=\{I_k,x_i\}$, $k=1,\dots,N $, etc.
\end{Proposition}

\begin{proof}[Idea of the proof] The Abel equations (\ref{Abel}) are the consequences of the general Abel-type equations (\ref{abel0}) and are calculated using the explicit form of the equations of separation in the considered case.
\end{proof}

\begin{Remark}\label{remark9}
Let us remark that the equations (\ref{Abel}) are easily integrated and the corresponding Abel--Jacobi problem is resolved in terms of the elementary functions.
\end{Remark}

\subsection{Non-standard variable separation: special degenerated case}\label{section4.4}
\subsubsection{Separated variables and equation of separation}\label{section4.4.1}
Let us now consider special degenerated case $c_{12}=0$, $c_{22}=0$.
 Although this case is a partial case of the above construction, but it is so special that deserves a separate consideration.

The following corollary of the Theorem \ref{mainteo} holds true.

\begin{Corollary} Let us assume that the matrix $C$ is such that $c_{11}\neq 0$, $c_{12}=c_{22}=0$.
Let the coordinates $x_i$, $p_i$, $i=1,\dots,N $ be defined by \eqref{p-qnonst}, where $A(u)=L_{11}(u)$, $B(u)=L_{21}(u)$ and the Lax matrix $L(u)$ be defined by the formulae \eqref{prodLraz}, \eqref{LCraz}. Then
the equations of separation have the following form
\begin{gather*}%\label{eqsepnst1'}
 A(x_i)=I(x_i)=0, \qquad i=1,\dots,N .
\end{gather*}
\end{Corollary}

\begin{proof} At first we observe that in the case $c_{12}=c_{22}=0$ we have that $\det (C)=0$.
 As it follows from the equations (\ref{eqsepnst1}) in the case $c_{12}=0$ the equation of separation are written simply as follows
 \begin{gather*}
 I(x_i)=0, \qquad i=1,\dots,N .
\end{gather*}
But, taking into account that $c_{22}=0$ we will obtain from~(\ref{Iu})
\begin{gather*}
I(u)=A(u).
\end{gather*}
The corollary is proven.\end{proof}

\begin{Remark}\label{remark10} Observe that in the considered special degenerated case the coordinates~$x_i$ are the action variables and are expressed via the integrals of motion.
\end{Remark}

\subsubsection{The Abel equations}\label{section4.4.2}
Let us now consider the Abel-type equations. In this case, due to the fact that the coordinates~$x_i$ are expressed via the integrals, the Abel-type equations should be written for momenta.

The following proposition holds true.

\begin{Proposition}
The equations of motion for the momenta of separation $p_i$ described in the previous subsection are written in the Abel-type form as follows
\begin{gather}\label{Abel1}
\sum\limits_{i=1}^N \dfrac{x_i^{n-j}}{ \sum\limits_{l=1}^{N-1} (N-l) x_i^{N-l-1}I_l}
\frac{1}{p_i}\frac{{\rm d}p_i}{{\rm d}t_k} =\d^j_{k}, \qquad j,k=1,\dots,N ,
\end{gather}
where $\frac{{\rm d}p_i}{{\rm d}t_k}=\{I_k,x_i\}$, $k=1,\dots,N $ and $x_i$ are expressed in terms of the integrals of motion as solutions of the equations of separation $I(x_i)=0$, $i =1,\dots,N$.
\end{Proposition}

\begin{proof}[Idea of the proof] The Abel equations (\ref{Abel1}) are the consequences of the general Abel-type equations (\ref{abelMom0}) and are calculated using the explicit form of the equations of separation in the considered case.
\end{proof}

Introducing the angle variables $\phi_i$ by
 \begin{gather*}p_i=\exp{\phi_i}, \qquad i=1,\dots,N \end{gather*}
we obtain that the equations (\ref{Abel1}) are easily integrated as follows
\begin{gather*}%\label{Abel1}
\sum\limits_{i=1}^N \dfrac{x_i^{N-j} }{ \sum\limits_{l=1}^{N-1} (N-l) x_i^{N-l-1}I_l} \big(\phi_i-\phi_i^0\big) = \big(t_j-t_j^0\big), \qquad j=1,\dots,N .
\end{gather*}
The last equations have the following explicit solution
\begin{gather*}%\label{Abel1}
 \phi_i = \phi_i^0+ \sum\limits_{j=1} \big(M(I)^{-1}\big)_{ij}\big(t_j-t_j^0\big), \qquad \text{where}\qquad M(I)_{ij}= \frac{x_i^{N-j}}{\sum\limits_{l=1}^{N-1} (N-l)x_i^{N-l-1} I_l},
\end{gather*}
and we remind that $x_i$ are functions of the integrals of motion.

To summarize: the case $c_{12}=c_{22}=0$ admits immediate explicit construction of the action-angle variables in the corresponding XXX spin chain.
\subsection[$N=2$ example]{$\boldsymbol{N=2}$ example}\label{section4.5}
\subsubsection{The model}\label{section4.5.1}
Let us consider the simplest non-trivial example of $N=2$. We will have
\begin{gather*}%\label{prodLraz2}
L^{(1,2)}(u)= L^{(\nu_1)}(u) L^{(\nu_2)}(u),
\end{gather*}
 and the basic one-spin matrices $L^{(i)}(u)$ are written as follows
\begin{gather*}%\label{onepoleLraz}
L^{(\nu_k)}(u)=1+\frac{1}{u-\nu_k} \sum\limits_{i,j=1}^2 S^{(k)}_{ji} X_{ij}.
\end{gather*}
Hereafter for the purpose of convenience we will locate the poles as follows: $\nu_1=1$, $\nu_2=-1$, and set
 $S^{(1)}_{ij}=S_{ij}$, $S^{(2)}_{kl}=T_{kl}$. The Poisson brackets are those of $\mathfrak{gl}(2)\oplus \mathfrak{gl}(2)$
 \begin{subequations}\label{gl2plgl2}
\begin{gather}
\{{S}_{ij}, {S}_{kl}\}= (\d_{kj}{S}_{il}- \d_{il} {S}_{kj}),\\
\{{T}_{ij}, {T}_{kl}\}= (\d_{kj}{T}_{il}- \d_{il} {T}_{kj}),\\
\{{T}_{ij}, {S}_{kl}\}=0.
\end{gather}
\end{subequations}
As previously we will be interested in the Lax matrices of the following form:
\begin{gather*}
L(u)= L^{(1,2}(u) C,
\end{gather*}
where the matrix $C$ is an arbitrary constant degenerated matrix $C=\sum\limits_{i,j=1}^2 c_{ij} X_{ij}$, $\det (C)=0$.

The generating function of integrals of motion is
\begin{gather*}
I(u)= \operatorname{tr} L(u).
\end{gather*}
More explicitly, we will have that
\begin{gather*}
I(u)=\frac{1}{ (1-u^2)} \big((c_{11}+c_{22}) u^{2}+ I_1 u+ I_2\big) ,
\end{gather*}
where the Hamiltonians $I_k$ are non-homogeneous polynomials in the dynamical variables
\begin{gather*}
I_1=- \sum\limits_{i,j=1}^2 c_{ij}( S_{ij}+ T_{ij} ),\\
I_2= \sum\limits_{i,j=1}^2 c_{ij} \sum\limits_{k=1}^2 T_{ik}S_{kj}- \sum\limits_{i,j=1}^2 c_{ij} ( S_{ij}-T_{ij}) - \sum\limits_{i=1}^2 c_{ii} .
\end{gather*}

\subsubsection{Separated variables}\label{section4.5.2}
Let, as previously, the functions $A(u)$ and $B(u)$ are defined as follows
\begin{gather*}
A(u)=L_{11}(u),\qquad B(u)=L_{21}(u).
\end{gather*}
For the considered $N=2$ case of the Lax matrix we obtain
\begin{gather*}
A(u)= \frac{1}{ (1-u^2)}\big( c_{11} u^{2}-( (S_{11}+T_{11})c_{11}+(T_{21}+S_{21})c_{21}) u \\
\hphantom{A(u)=}{} + (T_{11}-S_{11}+S_{11}T_{11}+S_{21}T_{12}-1)c_{11}+(T_{21}-S_{21}+S_{21}T_{22}+T_{21}S_{11})c_{21}\big),
\\
B(u)= \frac{1}{(1-u^2)}\big( c_{21} u^{2}- ( c_{11}(T_{12}+S_{12})-c_{21}(T_{22}+S_{22})) u \\
\hphantom{B(u)=}{} + (S_{12}T_{11}+T_{12}S_{22}+T_{12}-S_{12}+)c_{11}+(S_{12}T_{21}-1+S_{22}T_{22}+T_{22}-S_{22})c_{21}\big).
\end{gather*}
We will again define the (quasi)canonical coordinates as follows
\begin{gather*}
A(x_i)=0, \qquad p_i=B(x_i),
\end{gather*}
that is $x_1$, $x_2$ are the solutions of the following quadratic equation
\begin{gather*}
 c_{11} x^{2}-( (S_{11}+T_{11})c_{11}+(T_{21}+S_{21})c_{21}) x \\ \qquad{}+(T_{11}-S_{11}+S_{11}T_{11}+S_{21}T_{12}-1)c_{11}+(T_{21}-S_{21}+S_{21}T_{22}+T_{21}S_{11})c_{21}=0
\end{gather*}
and
\begin{gather*}
p_i= \frac{1}{(1-x_i^2)}\big( c_{21} x_i^{2}- ( c_{11}(T_{12}+S_{12})-c_{21}(T_{22}+S_{22})) x_i \\
\hphantom{p_i=}{} + (S_{12}T_{11}+T_{12}S_{22}+T_{12}-S_{12}+)c_{11}+(S_{12}T_{21}-1+S_{22}T_{22}+T_{22}-S_{22})c_{21}\big).
\end{gather*}
By virtue of our general theory the coordinates $x_i$, $p_i$ are quasi-canonical
\begin{gather}\label{quascan}
 \{x_i, x_j\}=0, \qquad \{p_i, p_j\}=0 , \qquad \{ x_i, p_i\}=p_i, \qquad \forall\, i,j = 1,2,
\end{gather}
and satisfy the following equation of separation
\begin{gather*}
\frac{c_{12}}{c_{11}} p_i=\frac{1}{ (1-x_i^2)} \big((c_{11}+c_{22}) x_i^{2}+ I_1 x_i+ I_2\big), \qquad i= 1,2.
\end{gather*}
In the case $c_{12}=0$, $c_{22}=0$ these equations of separation degenerate to the form
\begin{gather}\label{degsepeq2}
 \big(c_{11} x_i^{2}+ I_1 x_i+ I_2\big)=0, \qquad i= 1,2,
\end{gather}
i.e., the coordinates of separation $x_i$ became the action coordinates in this case.
\subsubsection{The reconstruction formulae}\label{section4.5.3}
In order to completely resolve the problem in the case $N=2$ it is necessary to be able to reconstruct the original dynamical variables from the separated variables and Casimir functions, i.e., it is necessary to solve the following system of eight linear-quadratic equations
%\begin{subequations}\label{rec1}
\begin{gather*}
c_{11} x_i^2+((-S_{11}-T_{11}) c_{11}-(S_{21}+T_{21}) c_{21}) x_i+(-1-S_{11}+T_{11}+S_{11} T_{11}+S_{21} T_{12}) c_{11}\nonumber\\
\qquad{} +(T_{21}+S_{21} T_{22}+T_{21} S_{11}-S_{21}) c_{21}=0, \qquad i = 1,2,
\\
\big(x_i^2-1\big)p_i=(c_{21} x_i^2+((-T_{12}-S_{12}) c_{11}+(-T_{22}-S_{22}) c_{21}) x_i\nonumber\\
\qquad{} +(-S_{12}+T_{12} S_{22}+S_{12} T_{11}+T_{12}) c_{11}\nonumber\\
\qquad{} + (T_{22}-S_{22}+S_{22} T_{22}+S_{12} T_{21}-1) c_{21}), \qquad i= 1,2,
\\
C_1=S_{11}S_{22}-S_{12}S_{21},
\qquad
C_2=T_{11}T_{22}-T_{12}T_{21},
\qquad
c_1=S_{11}+S_{22},
\qquad
c_2=T_{11}+T_{22}
\end{gather*}
on the eight variables $S_{ij}$, $T_{ij}$, $i,j= 1,2$.

Hereafter for the purpose of simplicity we will put $c_{21}=0$, $c_{22}=0$ and $c_1=c_2=0$.

The following theorem holds true.
\begin{Theorem}\quad
\begin{enumerate}\itemsep=0pt
\item[$(i)$] The mechanical coordinates $S_{ij}$, $T_{ij}$ are parametrized by the separated coordinates $x_i$, $p_i$ and the values of the Casimir functions $C_1$, $C_2$ as follows
\begin{subequations}\label{recons1}
\begin{gather}
S_{11}=\bigl(((x_1+1) (x_1-1)^2 p_1-(1+x_2) (-1+x_2)^2 p_2) C_1\nonumber\\
\hphantom{S_{11}=}{}+(-1+x_2)^2 (x_1+1) (x_1-1)^2 p_1-(1+x_2)
(-1+x_2)^2 (x_1-1)^2 p_2\bigr)\nonumber\\
\hphantom{S_{11}=}{} \times \bigl(((x_1-1) (x_1+1) p_1-(-1+x_2) (1+x_2) p_2) C_1\nonumber\\
\hphantom{S_{11}=}{}+(-1+x_2)^2 (x_1-1) (x_1+1) p_1-(-1+x_2) (1+x_2) (x_1-1)^2 p_2\bigr)^{-1},\\
S_{12}=
c_{11} \bigl((x_1-x_2) C_1^2+(x_1-x_2) (x_1^2-2 x_1+x_2^2+2-2 x_2) C_1\nonumber\\
\hphantom{S_{12}=}{} +(x_2-1)^2 (x_1-1)^2 (x_1-x_2)\bigr)
\bigl(((x_1-1) (x_1+1) p_1-(x_2-1) (1+x_2) p_2) C_1\nonumber\\
\hphantom{S_{12}=}{} +(x_2-1)^2 (x_1-1) (x_1+1) p_1-(x_2-1) (1+x_2) (x_1-1)^2 p_2\bigr)^{-1},
\\
S_{22}=-S_{11}, \qquad S_{21}=\frac{(S_{11}S_{22}-C_1)}{S_{12}},
\\
T_{11}= \bigl(((x_1-1) (x_1+1) (1+x_2) p_1-(x_2-1) (1+x_2) (x_1+1) p_2) C_1\nonumber\\
\hphantom{T_{11}=}{} +(1+x_2) (x_2-1)^2 (x_1-1) (x_1+1) p_1
-(x_2-1) (1+x_2) (x_1+1) (x_1-1)^2 p_2\bigr)\nonumber\\
\hphantom{T_{11}=}{}\times
\bigl(((x_1-1) (x_1+1) p_1-(x_1-1) (1+x_2) p_2) C_1\nonumber\\
\hphantom{T_{11}=}{}+(x_2-1)^2 (x_1-1) (x_1+1) p_1-(x_2-1) (1+x_2) (x_1-1)^2 p_2\bigr)^{-1},\\
T_{12}= \frac{ p_1 p_2 }{c_{11}}(x_2-1) (1+x_2) (x_1-1) (x_1+1) (x_1-x_2)\bigl( ((x_1-1) (x_1+1) p_1\nonumber\\
\hphantom{T_{12}=}{} -(x_2-1) (1+x_2) p_2) C_1 +(x_2-1)^2 (x_1-1) (x_1+1) p_1\nonumber\\
\hphantom{T_{12}=}{} -(x_2-1) (1+x_2) (x_1-1)^2 p_2)\bigr)^{-1},
\\
T_{22}=-T_{11}, \qquad T_{21}=\frac{(T_{11}T_{22}-C_2)}{T_{12}}.
\end{gather}
\end{subequations}

\item[$(ii)$] If the Poisson relations among $p_i$ and $q_j$ are quasi-canonical~\eqref{quascan} then the variables~$S_{ij}$,~$T_{ij}$ given by~\eqref{recons1} satisfy the Poisson brackets~\eqref{gl2plgl2}.
\end{enumerate}
\end{Theorem}

\begin{proof}[Idea of the proof] Item (i) of the theorem is proven by direct calculations. The item~(ii) follows from our general theory and can be also checked by direct calculations.
\end{proof}

\begin{Remark}\label{remark11} Using the formulae (\ref{recons1}) one can obtain the expression for the integrals~$I_1$ and~$I_2$ in terms of the separated variables
\begin{subequations}\label{ham}
\begin{gather}
I_1=\frac{c_{12}(x_1-1)(x_1+1)p_1}{c_{11}(x_1-x_2)}-\frac{c_{12}(x_2-1)(1+x_2) p_2}{c_{11}(x_1-x_2)}-(x_1+x_2)c_{11},
\\
I_2 = \frac{-c_{12} x_2 (x_1-1) (x_1+1) p_1}{c_{11} (x_1-x_2)}+\frac{c_{12} x_1 (x_2-1) (1+x_2) p_2}{c_{11} (x_1-x_2)}+x_1 x_2 c_{11}.
\end{gather}
\end{subequations}
\end{Remark}

\subsubsection{Abel-type equations: general degenerated case}\label{section4.5.4}
Using either our general theory or the explicit form of the Hamiltonians in terms of the separated variables (\ref{ham}) one easily derives the Abel-type equations. In case $c_{12}\neq 0$ the Abel equations are written for the separated coordinates $x_i$ as follows
\begin{gather*}
\frac{x_1 }{(c_{11}+c_{22}) x_1^{2}+ I_1 x_1+ I_2}\frac{{\rm d} x_1}{{\rm d}t_1} +\frac{x_2}{ (c_{11}+c_{22}) x_2^{2}+ I_1 x_2+ I_2} \frac{{\rm d} x_2}{{\rm d}t_1}=1,
\\
\frac{x_1 }{(c_{11}+c_{22}) x_1^{2}+ I_1 x_1+ I_2}\frac{{\rm d} x_1}{{\rm d}t_2} +\frac{x_2}{ (c_{11}+c_{22}) x_2^{2}+ I_1 x_2+ I_2} \frac{{\rm d} x_2}{{\rm d}t_2}=0,
\\
\frac{1 }{(c_{11}+c_{22}) x_1^{2}+ I_1 x_1+ I_2}\frac{{\rm d} x_1}{{\rm d}t_1} +\frac{1}{ (c_{11}+c_{22}) x_2^{2}+ I_1 x_2+ I_2} \frac{{\rm d} x_2}{{\rm d}t_1}=0,
\\
\frac{1 }{(c_{11}+c_{22}) x_1^{2}+ I_1 x_1+ I_2}\frac{{\rm d} x_1}{{\rm d}t_2} +\frac{1}{ (c_{11}+c_{22}) x_2^{2}+ I_1 x_2+ I_2} \frac{{\rm d} x_2}{{\rm d}t_2}=1.
\end{gather*}
They are easily resolved in terms of elementary functions.

\subsubsection{Abel-type equations: special degenerated case}\label{section4.5.5}
In case $c_{12}= 0$, $c_{22}=0$, $c_{11}\neq 0$ the Abel equations written for the separated coordinates $x_i$ are trivial due to the fact that $x_i$ in this case are expressed via the integrals of motion. The non-trivial Abel-type equations are written for the conjugated momenta~$p_i$
\begin{gather*}
\frac{x_1}{ (x_1-x_2)p_1}\frac{{\rm d}p_1}{{\rm d}t_1}+\frac{ x_2}{ (x_2-x_1)p_2}\frac{{\rm d}p_2}{{\rm d}t_1}=c_{11},\\
\frac{1}{ (x_1-x_2)p_1}\frac{{\rm d}p_1}{{\rm d}t_1}+\frac{ 1}{ (x_2-x_1)p_2}\frac{{\rm d}p_2}{{\rm d}t_1}=0,\\
\frac{x_1}{ (x_1-x_2)p_1}\frac{{\rm d} p_1}{{\rm d}t_2}+\frac{ x_2}{ (x_2-x_1)p_2}\frac{{\rm d} p_2}{{\rm d}t_2}=0,\\
\frac{1}{ (x_1-x_2)p_1}\frac{{\rm d} p_1}{{\rm d}t_2}+\frac{ 1}{ (x_2-x_1)p_2}\frac{{\rm d} p_2}{{\rm d}t_2}=c_{11},
\end{gather*}
where we have used the following relations among $x_i$ and $I_i$
\begin{gather*}
I_1 = -c_{11}(x_1+x_2), \qquad I_2= c_{11} x_1 x_2.
\end{gather*}

Introducing the angle variables $p_i=\exp{\phi_i}$, $i= 1,2$ we obtain that
\begin{gather*}
 \phi_1 = \phi_1^0+ c_{11}\big(\big(t_1-t_1^0\big)-x_2 \big(t_2-t_2^0\big)\big),
\qquad
 \phi_2 = \phi_2^0+ c_{11}\big(\big(t_1-t_1^0\big)-x_1 \big(t_2-t_2^0\big)\big),
\end{gather*}
where $x_i$, $i= 1,2$ are the constants of motion calculated from the equations~(\ref{degsepeq2}).

\section{Classical XXZ spin model}\label{section5}
\subsection[The $r$-matrix, the basic Lax matrix and quadratic brackets]{The $\boldsymbol{r}$-matrix, the basic Lax matrix and quadratic brackets}\label{section5.1}
Let us consider the following $r$-matrix
\begin{gather}
r^{12}(u_1,u_2)=\frac{1}{2}\frac{u_1+u_2}{u_1 - u_2}(X_{11}
\otimes X_{11}+ X_{22}
\otimes X_{22}) \nonumber\\
\hphantom{r^{12}(u_1,u_2)=}{} + \frac{u_1}{u_1 - u_2} X_{12} \otimes X_{21} + \frac{u_2}{u_1 - u_2} X_{21}
\otimes X_{12}.\label{drmtrig}
\end{gather}
It is the so-called standard trigonometric $r$-matrix \cite{BD} in the special gauge \cite{Jurco}. We are interested in obtaining of the Lax matrices satisfying the algebra~(\ref{rmbr}) for this $r$-matrix.

 The following proposition holds true~\cite{SkrDub2}.

\begin{Proposition}\label{algSkltrig} Let $\nu \neq 0$, $\nu \neq \infty$. Then
the Lax matrix $L^{(\nu)}(u)$ defined by the formula
\begin{gather}\label{LaxTrigSkl}
L^{(\nu)}(u)=\sum\limits_{i=1}^2 S_{0i}X_{ii} + \dfrac{1}{2}\frac{\nu+u}{\nu - u} \sum\limits_{i=1}^2 S_{ii}X_{ii} +
 \frac{u}{\nu - u}S_{21}
 X_{12} + \frac{\nu}{\nu - u} S_{12}
 X_{21},
\end{gather}
where the Poisson brackets among the coordinate functions $S_{0i}$, $S_{ii}$, $S_{ij}$ are the following
\begin{subequations}\label{trigSkl2}
\begin{gather}
\{S_{0i}, S_{0j}\}=0, \qquad \{S_{0i}, S_{jj}\} =0, \qquad \{S_{ii}, S_{jj}\} =0, \qquad i,j= 1,2,
\\
\{S_{0i}, S_{ij}\}= \frac{1}{4}S_{ii} S_{ij}, \qquad \{S_{ii}, S_{ij}\} = S_{0i}S_{ij} , \qquad i,j= 1,2,
\\
\{S_{0i}, S_{ji}\}= - \frac{1}{4}S_{ii} S_{ji}, \qquad \{S_{ii}, S_{ji}\} = -S_{0i}S_{ji} , \qquad i,j= 1,2,
\\
\{S_{ij}, S_{ji}\}= S_{0j}S_{ii}-S_{0i}S_{jj} , \qquad i,j= 1,2
\end{gather}
\end{subequations}
satisfies Poisson brackets~\eqref{rmbr}.
\end{Proposition}

Let us describe the central elements (Casimir functions) of the quadratic Poisson algebra~(\ref{trigSkl2}). As it is well-known (see, e.g.,~\cite{SkrRoub} for the detailed proof), the Casimir functions of any quadratic Poisson algebra are obtained by {expanding} in powers of~$u$ of the following generating function
\begin{gather*}
C(u)=\det (L(u)).
\end{gather*}
More explicitly, we will have the following {expansion}, yielding three Casimir functions~$C_k$
\begin{gather*}
C(u)=\frac{1}{4(\nu-u)^2} \sum\limits_{k=0}^2 u^k C_{2-k},
\end{gather*}
where
\begin{gather*}
C_0=\prod\limits_{i=1}^2(2S_{0i}-S_{ii}), \qquad C_1= 2(S_{11}S_{22}-2S_{12}S_{21}-4S_{01}S_{02}), \qquad C_2=\prod\limits_{i=1}^2(2S_{0i}+S_{ii}).
\end{gather*}

Besides, there also exist two additional quadratic Casimir functions
\begin{gather*}
 c_{i}= 4S_{0i}^2-S_{ii}^2, \qquad i= 1,2.
\end{gather*}
The functions $c_1$, $c_2$, $C_1$, $C_2$ are functionally independent. That is why the dimension of the generic symplectic leaf of the brackets~(\ref{trigSkl2}) is two.

\begin{Remark}\label{remark12}Observe, that there exists also the following quadratic Casimir function
\begin{gather*}
C_2' = 4(S_{01}-S_{02})^2-(S_{11}+S_{22})^2.
\end{gather*}
On its zero level surface one can put
\begin{gather*}%\label{redSkl}
S_{01}=S_{02}, \qquad S_{11}=-S_{22}
\end{gather*}
and arrive exactly to the trigonometric degeneration of the famous elliptic Sklyanin algebra on the space $\mathfrak{sl}(2)$ extended with the help of one-dimensional center $S_{0}=S_{01}=S_{02}$.
\end{Remark}

\subsection{ The Lax matrix and integrals of the classical XXZ model}\label{section5.2}
Following the general theory of the quadratic $r$-matrix brackets it is possible to consider the product of Lax matrices
\begin{gather}\label{ProdTrigN}
L^{(1,2,\dots,N)}(u)= {L}^{(\nu_1)}(u)\cdots {L}^{(\nu_N)}(u),
\end{gather}
where the elementary Lax matrices $L^{\nu_k}(u)$ are given by the formula (\ref{LaxTrigSkl})
\begin{gather*}
L^{(\nu_k)}(u)=\sum\limits_{i=1}^2 S^{(k)}_{0i}X_{ii} + \dfrac{1}{2}\frac{\nu_k+u}{\nu_k - u} \sum\limits_{i=1}^2 S^{(k)}_{ii}X_{ii} +
 \frac{u}{\nu_k - u}S^{(k)}_{21}
 X_{12} + \frac{\nu_k}{\nu_k - u} S^{(k)}_{12}
 X_{21}.
\end{gather*}
The Poisson brackets among the functions $S^{(k)}_{0i}$, $S^{(k)}_{ij}$ exactly repeat the Poisson brackets (\ref{trigSkl2}). The Poisson brackets among the functions $S^{(m)}_{0i}$, $S^{(m)}_{ij}$ and $S^{(n)}_{0k}$, $S^{(n)}_{kl}$ are trivial
\begin{gather*}
\big\{S^{(m)}_{0i}, S^{(n)}_{0j}\big\}=0, \qquad \big\{S^{(m)}_{0i}, S^{(n)}_{kl}\big\} =0, \qquad \big\{S^{(m)}_{ij}, S^{(n)}_{kl}\big\} =0, \\ \forall\, i,j,k,l= 1,2, \qquad m, n =1,\dots,N ,\qquad m\neq n.
\end{gather*}

The independent Casimir functions are the functions
\begin{gather*}
 C^{(m)}_1= 2\big(S^{(m)}_{11}S^{(m)}_{22}-2S^{(m)}_{12}S^{(m)}_{21}-4S^{(m)}_{01}S^{(m)}_{02}\big), \\ C^{(m)}_2=\prod\limits_{i=1}^2\big(2S^{(m)}_{0i}+S^{(m)}_{ii}\big), \qquad m=1,\dots,N.
\end{gather*}
Besides there exist also the following quadratic Casimir functions
\begin{gather}\label{cmi}
 c^{(m)}_{i}= 4\big(S^{(m)}_{0i}\big)^2-\big(S^{(m)}_{ii}\big)^2, \qquad i= 1,2, \qquad m=1,\dots,N .
\end{gather}
The dimension of the generic symplectic leaf is~$2N$. So we have to have~$N$ independent integrals for the complete integrability of the Hamiltonian system on the generic symplectic leaf. We will obtain them from the following Lax matrix
\begin{gather}\label{ProdTrigC}
L(u)=L^{(1,2,\dots,N)}(u)C,
\end{gather}
where $C$ is a constant matrix, in a standard way
\begin{gather*}
I(u)=\operatorname{tr}(L(u)).
\end{gather*}
The matrix $C$ is not arbitrary. It should satisfy the following condition
\begin{gather}\label{bound}
[r(u,v), C\otimes C]=0.
\end{gather}

 The following proposition is proven by the direct calculation.
\begin{Proposition}
The matrix $C=\sum\limits_{i,j=1}^2 c_{ij}X_{ij}$ satisfy the condition \eqref{bound} for the case of the trigonometric $r$-matrix~\eqref{drmtrig} if $C$ is diagonal, i.e., $c_{12}=c_{21}=0$.
\end{Proposition}

The generating function of the integrals of motion is written in the case of the diagonal matrices $C$ as follows
\begin{gather*}
I(u)= \frac{1}{\prod\limits_{i=1}^N (u-\nu_i)} \sum\limits_{k=0}^{N} u^{N-k} I_{k},
\end{gather*}
where
\begin{gather*}
I_0= c_{11}\prod\limits_{k=1}^N \big(2S^{(k)}_{01}-S^{(k)}_{11}\big)+ c_{22}\prod\limits_{k=1}^N \big(2S^{(k)}_{02}-S^{(k)}_{22}\big),
\\
I_1= (-1)^N c_{11}\left(\sum\limits_{l=1}^N \big(2S^{(l)}_{01}-S^{(l)}_{11}\big)\prod\limits_{k=1, k\neq l}^N \nu_k \big(2S^{(k)}_{01}+S^{(k)}_{11}\big)\right.\\
\left. \hphantom{I_1=}{} -
\sum\limits_{l,m=1,l<m}^N 4\nu_m S^{(m)}_{12} S^{(l)}_{21} \prod\limits_{k=1, k\neq l,m}^N \nu_k \big(2S^{(k)}_{01}+S^{(k)}_{11}\big)\right)\\
\hphantom{I_1=}{}+
(-1)^N c_{22}\left(\sum\limits_{l=1}^N \big(2S^{(l)}_{02}-S^{(l)}_{22}\big)\prod\limits_{k=1, k\neq l}^N \nu_k \big(2S^{(k)}_{02}+S^{(k)}_{22}\big)\right.\\
\left.\hphantom{I_1=}{} -
\sum\limits_{l,m=1,l<m}^N 4\nu_l S^{(m)}_{21} S^{(l)}_{12} \prod\limits_{k=1, k\neq l,m}^N \nu_k \big(2S^{(k)}_{02}+S^{(k)}_{22}\big)\right),
\\
\cdots\cdots\cdots\cdots\cdots\cdots\cdots\cdots\cdots\cdots\cdots\cdots\cdots\cdots\cdots\cdots\cdots
\cdots\cdots\cdots\cdots\cdots\cdots
\\
I_N= (-1)^N \left(\prod\limits_{l=1}^N \nu_l\right)\left(c_{11}\prod\limits_{k=1}^N \big(2S^{(k)}_{01}+S^{(k)}_{11}\big)+ c_{22}\prod\limits_{k=1}^N \big(2S^{(k)}_{02}+S^{(k)}_{22}\big)\right),
\end{gather*}
etc. It is possible to show that the functions $I_1,\dots,I_N$ are functionally independent on the generic symplectic leaf, while the function $I_0$ is dependent on the Casimir functions and other integrals.
In particular, in the relevant case $c_{22}=0$ we will have
\begin{gather}\label{i0}
I_0=\frac{(-1)^N c_{11}^2 \prod\limits_{l=1}^N \nu_l c^{(1)}_l}{I_N},
\end{gather}
where $c^{(1)}_l$ are the above-defined Casimir functions (\ref{cmi}).

\subsection{Non-standard variable separation}\label{section5.3}
\subsubsection{Separated variables and equation of separation}\label{section5.3.1}
The trigonometric $r$-matrix (\ref{drmtrig}) evidently satisfies the conditions (\ref{condrms}). Hence, using our inverted definition (\ref{p-qnonst}) one can define the (quasi)canonical coordinates as follows
\begin{gather*}
A(x_i)=0, \qquad p_i=B(x_i),
\end{gather*}
where the functions $A(u)$ and $B(u)$ are standardly defined using the Lax matrix $L(u)$
\begin{gather*}
A(u)=L_{11}(u),\qquad B(u)=L_{21}(u).
\end{gather*}

For the Lax operator defined by the formulae (\ref{ProdTrigN}), (\ref{ProdTrigC}) we also obtain
\begin{gather*}
A(u)= \frac{1}{\prod\limits_{i=1}^N (u-\nu_i)}\left( \sum\limits_{k=0}^{N} A_k u^{N-k}\right),
\qquad
B(u)= \frac{1}{\prod\limits_{i=1}^N (u-\nu_i)}\left( \sum\limits_{k=0}^{N} B_k u^{N-k}\right),
\end{gather*}
where $A_k$, $B_k$ are the functions of the dynamical variables $S^{(l)}_{0j}$, $S^{(l)}_{ij}$ and the constants $c_{ij}$.

In particular, we have that
\begin{gather*}
A_0= c_{11}\prod\limits_{k=1}^N \big(2S^{(k)}_{01}-S^{(k)}_{11}\big), \qquad
A_N=(-1)^N c_{11} \left(\prod\limits_{l=1}^N \nu_l\right)\prod\limits_{k=1}^N \big(2S^{(k)}_{01}+S^{(k)}_{11}\big).
\end{gather*}
It is important that on a generic symplectic leaf one has $A_0\neq 0$, $A_N\neq 0$, i.e., the separating polynomial $A(u)$ has degree $N$ and provide the needed number of the coordinates of separation.

The following theorem holds true.
\begin{Theorem}\label{mainteo2} Let $C=\operatorname{diag}(c_{11},c_{22})$ and $c_{11}\neq 0$.\footnote{In the case $c_{11}=0$, $c_{22}\neq 0$ one redefines the separating function putting $A(u)=L_{22}(u)$, $B(u)=L_{12}(u).$} Let the coordinates $x_i$, $p_i$, $i=1,\dots,N $ be defined by \eqref{p-qnonst}, where $A(u)=L_{11}(u)$, $B(u)=L_{21}(u)$ and the Lax matrix $L(u)$ be defined by the formulae \eqref{ProdTrigN}, \eqref{ProdTrigC}. Then the coordinates $x_i$, $p_i$ are separated coordinates for the integrable system with the algebra of integrals of motion generated by $I(u)=\operatorname{tr} L(u)$ if and only if $c_{22}=0$. The equations of separation in this case have the following form
\begin{gather*}%\label{eqsepnst3}
 I(x_i)=A(x_i)=0, \qquad i=1,\dots,N ,
\end{gather*}
i.e., the coordinates $x_i$ coincide with action variables for the system.
\end{Theorem}

\begin{proof}[Sketch of the proof] Proof of the theorem is similar to the proof of the corresponding theorem in the case of the rational $r$-matrix. The only difference is the algebra of separating functions that reads in this case as follows
\begin{gather*}
\{B(u), B(v) \}=0,
\\
\{A(u), B(v) \}=-\frac{1}{2}\frac{u+v}{u-v}A(u)B(v)+\frac{v}{u-v}A(v)B(u),
\\
\{A(u), A(v) \}=0.
\end{gather*}
By the virtue of this algebra we obtain in the trigonometric case the quasi-canonical Poisson brackets of the following form
\begin{gather*}%\label{quascan2}
 \{x_i, x_j\}=0, \qquad \{p_i, p_j\}=0 , \qquad \{ x_i, p_i\}=x_i p_i, \qquad \forall\, i,j = 1,2.
\end{gather*}
 The other specific feature of the trigonometric case is that we are always in the special degeneration settings, i.e., $c_{12}=c_{21}=0$ and the equation of separation is
\begin{gather*}
I(x_i)=0, \qquad i=1,\dots,N.
\end{gather*}
It holds true if and only if $\det (C)=c_{11}c_{22}=0$, i.e., $c_{22}=0$ for non-trivial $c_{11}$. In this case $I(u)=A(u)$. This proves the theorem.
\end{proof}

\subsubsection{The Abel equations}\label{section5.3.2}
Let us now consider the Abel-type equations. In this case, due to the fact that the coordinates $x_i$ are expressed via the integrals, the Abel-type equations should be written for momenta.

Similar to the case of the rational $r$-matrix, the following proposition holds true.

\begin{Proposition}The equations of motion for the momenta of separation $p_i$ described in the previous subsection are written in the Abel-type form as follows
\begin{gather}\label{Abel2}
\sum\limits_{i=1}^N\left( \dfrac{x_i^{N-j}}{ \sum\limits_{l=1}^{N-1} (N-l) x_i^{N-l-1} I_l} -\d_{jN} \dfrac{x_i^{N} \dfrac{(-1)^N c_{11}^2 \prod\limits_{l=1}^N \nu_l c^{(1)}_l}{I^2_N}}{ \sum\limits_{l=1}^{N-1} (N-l) x_i^{N-l-1} I_l} \right)
\frac{1}{x_ip_i}\frac{{\rm d}p_i}{{\rm d}t_k} =\d^j_{k},
\end{gather}
$j,k=1,\dots,N$, where $\frac{{\rm d}p_i}{{\rm d}t_k}=\{I_k,x_i\}$, $k=1,\dots,N $ and $x_i$ are expressed in terms of the integrals of motion as solutions of the equations of separation
$I(x_i)=0$, $i =1,\dots,N$.
\end{Proposition}

\begin{proof} The Abel equations (\ref{Abel1}) are the consequences of the general Abel-type equations~(\ref{abelMom0}) and are calculated using the explicit form of the equations of separation in the considered case.
\end{proof}

\begin{Remark}\label{remark13} The different form of the Abel equations~(\ref{Abel2}) with respect to those in the rational case is explained by the different form of the quasi-canonical Poisson brackets and the fact that the integral $I_0$ in the pencil~$I(u)$ is expressed via~$I_N$ as in~(\ref{i0}).

As in the rational case, introducing the angle variables
 \begin{gather*}p_i=\exp{\phi_i}, \qquad i=1,\dots,N \end{gather*}
we obtain that the equations (\ref{Abel2}) are easily integrated as follows
\begin{gather}\label{Abel2'}
\sum\limits_{i=1}^N\left( \dfrac{x_i^{N-j}}{ \sum\limits_{l=1}^{N-1} (N-l) x_i^{N-l-1} I_l} -\d_{jN} \dfrac{x_i^{N} \dfrac{(-1)^N c_{11}^2 \prod\limits_{l=1}^N \nu_l c^{(1)}_l}{I^2_N}}{ \sum\limits_{l=1}^{N-1} (N-l) x_i^{N-l-1} I_l} \right) \frac{\big(\phi_i-\phi_i^0\big)}{x_i} =\big(t_j-t_j^0\big),
\end{gather} $j=1,\dots,N$.
The explicit solution of (\ref{Abel2'}) is constructed in the same manner as in the rational case.
\end{Remark}

\subsection[$N=2$ example]{$\boldsymbol{N=2}$ example}\label{section5.4}
\subsubsection{The model}\label{section5.4.1}
Let us consider the simplest non-trivial example of $N=2$. We will have
\begin{gather*}%\label{prodLraz2'}
L^{(1,2)}(u)= L^{(\nu_1)}(u) L^{(\nu_2)}(u),
\end{gather*}
 and the basic one-spin matrices $L^{(\nu_k)}(u)$ are written as follows
\begin{gather*}
L^{(\nu_k)}(u)=\sum\limits_{i=1}^2 S^{(k)}_{0i}X_{ii} + \dfrac{1}{2}\frac{\nu_k+u}{\nu_k - u} \sum\limits_{i=1}^2 S^{(k)}_{ii}X_{ii}\\
\hphantom{L^{(\nu_k)}(u)=}{} + \frac{u}{\nu_k - u}S^{(k)}_{21} X_{12} + \frac{\nu_k}{\nu_k - u} S^{(k)}_{12} X_{21}, \qquad k= 1,2.
\end{gather*}
Hereafter for the purpose of convenience we will put $\nu_1=1$, $\nu_2=-1$,
 \begin{gather*}S^{(1)}_{ij}=S_{ij},\qquad S^{(1)}_{0j}=S_{0j},\qquad
S^{(2)}_{kl}=T_{kl},\qquad S^{(2)}_{0j}=T_{0j}.\end{gather*} The Poisson brackets are those of the direct sum of two quadratic algebras (\ref{trigSkl2})
 \begin{subequations}\label{trigSklST}
\begin{gather}
\{S_{0i}, S_{0j}\}=0, \qquad \{S_{0i}, S_{jj}\} =0, \qquad \{S_{ii}, S_{jj}\} =0, \qquad i,j= 1,2,
\\
\{S_{0i}, S_{ij}\}= \frac{1}{4}S_{ii} S_{ij}, \qquad \{S_{ii}, S_{ij}\} = S_{0i}S_{ij} , \qquad i,j= 1,2,
\\
\{S_{0i}, S_{ji}\}= - \frac{1}{4}S_{ii} S_{ji}, \qquad \{S_{ii}, S_{ji}\} = -S_{0i}S_{ji} , \qquad i,j= 1,2,
\\
\{S_{ij}, S_{ji}\}= S_{0j}S_{ii}-S_{0i}S_{jj} , \qquad i,j= 1,2,
\\
\{T_{0i}, T_{0j}\}=0, \qquad \{T_{0i}, T_{jj}\} =0, \qquad \{T_{ii}, T_{jj}\} =0, \qquad i,j= 1,2,
\\
\{T_{0i}, T_{ij}\}= \frac{1}{4}T_{ii} T_{ij}, \qquad \{T_{ii}, T_{ij}\} = T_{0i}T_{ij} , \qquad i,j= 1,2,
\\
\{T_{0i}, T_{ji}\}= - \frac{1}{4}T_{ii} T_{ji}, \qquad \{T_{ii}, T_{ji}\} = -T_{0i}T_{ji} , \qquad i,j= 1,2,
\\
\{T_{ij}, T_{ji}\}= T_{0j}T_{ii}-T_{0i}T_{jj} , \qquad i,j= 1,2,
\\
\{S_{ij}, T_{kl}\}=\{S_{0i},T_{0j}\}= \{S_{0i},T_{kl}\}= \{T_{0j}, S_{kl}\}=0.
\end{gather}
\end{subequations}
The Casimir functions of the quadratic Sklyanin-type parenthesis of the direct sum are
\begin{gather*}
 C^{(1)}_1= 2(S_{11}S_{22}-2S_{12}S_{21}-4S_{01}S_{02}), \qquad C^{(2)}_1= 2(T_{11}T_{22}-2T_{12}T_{21}-4T_{01}T_{02}),
\\
 C^{(1)}_2=(2S_{01}+S_{11})(2S_{02}+S_{22}), \qquad C^{(2)}_2=(2T_{01}+T_{11})(2T_{02}+T_{22}).
\end{gather*}
Besides, there exist also the following quadratic Casimir functions
\begin{gather*}
 c^{(1)}_{i}= 4(S_{0i})^2-(S_{ii})^2, \qquad c^{(2)}_{i}= 4(T_{0i})^2-(T_{ii})^2, \qquad i= 1,2.
\end{gather*}

As previously we will be interested in the degenerated Lax matrices of the following form
\begin{gather*}
L(u)= L^{(1,2)}(u) C,
\end{gather*}
where the matrix $C=\operatorname{diag}(c_{11},c_{22})$. The generating function of integrals of motion is
\begin{gather*}
I(u)= \operatorname{tr} L(u).
\end{gather*}
More explicitly, we will have that
\begin{gather*}
I(u)=\frac{1}{ u^2-1} \big(I_0 u^{2}+ I_1 u+ I_2\big),
\end{gather*}
where the Hamiltonians $I_k$ are written as follows
\begin{gather*}
I_0= (c_{11}(2S_{01}-S_{11})(2T_{01}-T_{11})+ c_{22} (2S_{02}-S_{22}) (2T_{02}-T_{22})),
\\
I_1= ( c_{11}\bigl(-(2S_{01}-S_{11}) (2T_{01}+T_{11})+
(2T_{01}-T_{11}) (2S_{01}+S_{11})+ 4 T_{12} S_{21}\bigr) \\
\hphantom{I_1=}{} + c_{22}\bigl( -(2S_{02}-S_{22}) (2T_{02}+T_{22})+
(2T_{02}-T_{22}) (2S_{02}+S_{22})-4 T_{21} S_{12}\bigr)),\\
I_2= - (c_{11} (2S_{01}+S_{11})(2T_{01}+T_{11})+ c_{22} (2S_{02}+S_{22})(2T_{02}+T_{22})).
\end{gather*}
On the four-dimensional level surface of the Casimir functions the integrals $I_1$ and $I_2$ are functionally independent.

\subsubsection{Separated variables}\label{section5.4.2}
Let, as previously, the functions $A(u)$ and $B(u)$ be defined as follows
\begin{gather*}
A(u)=L_{11}(u),\qquad B(u)=L_{21}(u).
\end{gather*}
For the considered $N=2$-case of the Lax matrix we obtain
\begin{gather*}
A(u)= \frac{c_{11}}{ u^2-1} \bigl( (2S_{01}-S_{11})(2T_{01}-T_{11}) u^{2}+ \bigl(
(2T_{01}-T_{11}) (2S_{01}+S_{11}) \\
\hphantom{A(u)=}{}-(2S_{01}-S_{11}) (2T_{01}+T_{11})+ 4 T_{12} S_{21}\bigr) u+ (2S_{01}+S_{11})(2T_{01}+T_{11})\bigr),\\
B(u)=\frac{((-2 T_{01}+T_{11}) S_{12}+(2 S_{02}-S_{22}) T_{12}) u-(2 T_{01}+T_{11}) S_{12}-(2 S_{02}+S_{22}) T_{12}}{2\big(u^2-1\big)}.
\end{gather*}
We define the (quasi)canonical coordinates in our non-standard manner
\begin{gather*}
A(x_i)=0, \qquad p_i=B(x_i),
\end{gather*}
i.e., $x_1$, $x_2$ are the solutions of the following quadratic equation
\begin{gather*}
 (2S_{01}-S_{11})(2T_{01}-T_{11}) x^{2}+ \bigl((2T_{01}-T_{11}) (2S_{01}+S_{11})\\
 \qquad{} -(2S_{01}-S_{11}) (2T_{01}+T_{11})+ 4 T_{12} S_{21}\bigr) x+ (2S_{01}+S_{11})(2T_{01}+T_{11})=0,
\end{gather*}
and
\begin{gather*}
p_i= \frac{((-2 T_{01}+T_{11}) S_{12}+(2 S_{02}-S_{22}) T_{12}) x_i-(2 T_{01}+T_{11}) S_{12}-(2 S_{02}+S_{22}) T_{12}}{2\big(x_i^2-1\big)},
\end{gather*}
$i= 1,2$.

By the virtue of our general theory the coordinates $x_i$, $p_i$ are quasi-canonical
\begin{gather}\label{quascan2'}
 \{x_i, x_j\}=0, \qquad \{p_i, p_j\}=0 , \qquad \{ x_i, p_i\}=x_ip_i, \qquad \forall\, i,j= 1,2.
\end{gather}
and (in the case $c_{22}=0$) satisfy the following equations of separation
\begin{gather}\label{degsepeq2'}
 \big(I_0 x_i^{2}+ I_1 x_i+ I_2\big)=0, \qquad i= 1,2,
\end{gather}
The coordinates of separation $x_i$ became the action coordinates in this case.

\subsubsection{The reconstruction formulae}\label{section5.4.3}
In order to completely resolve the problem in the case $N=2$ it is necessary to be able to reconstruct the dynamical variables via the separated coordinates and Casimir functions, i.e., it is necessary to resolve the following system of twelve linear-quadratic equations
\begin{subequations}\label{rec2}
\begin{gather}
 (2S_{01}-S_{11})(2T_{01}-T_{11}) x_i^{2}+ \bigl(
(2T_{01}-T_{11}) (2S_{01}+S_{11})\nonumber\\
\qquad{} -(2S_{01}-S_{11}) (2T_{01}+T_{11})+ 4 T_{12} S_{21}\bigr) x_i+(2S_{01}+S_{11})(2T_{01}+T_{11})=0,
\\
p_i= \frac{((-2 T_{01}+T_{11}) S_{12}+(2 S_{02}-S_{22}) T_{12}) x_i-(2 T_{01}+T_{11}) S_{12}-(2 S_{02}+S_{22}) T_{12}}{2\big(x_i^2-1\big)},\nonumber\\
\qquad i = 1,2,
\\
 C_1= 2(S_{11}S_{22}-2S_{12}S_{21}-4S_{01}S_{02}),
\\
 C_2^2=(2S_{01}+S_{11})(2S_{02}+S_{22}),
\\
 K_1= 2(T_{11}T_{22}-2T_{12}T_{21}-4T_{01}T_{02}),
\\
K_2^2=(2T_{01}+T_{11})(2T_{02}+T_{22}).
\\
 c_1= 4(S_{01})^2-(S_{11})^2,
\\
 c_2= 4(S_{02})^2-(S_{22})^2,
\\
 k_1= 4(T_{01})^2-(T_{11})^2,
\\
k_2=4(T_{02})^2-(T_{22})^2,
\end{gather}
\end{subequations}
 on the twelve variables $S_{0i}$, $S_{ij}$, $T_{0i}$ $T_{ij}$, $i,j= 1,2$. Here $C_i$, $K_i$, $c_i$, $k_i$, $i= 1,2$ are some fixed values of the Casimir functions.

For the purpose of simplicity we will hereafter consider the restriction of our quadratic Poisson algebra to the trigonometric Sklyanin algebra case, i.e., we will put \begin{gather*}S_{01}=S_{02},\qquad S_{22}=-S_{11}, \qquad T_{02}=T_{01}, \qquad T_{22}=-T_{11}.\end{gather*}
After such the reduction one can neglect the last four equations in the system~(\ref{rec2}) due to the fact that the Casimirs $c_i$, $k_i$ become dependent on the Casimirs $C_i$, $K_i$.

The following theorem holds true.
\begin{Theorem}\quad
\begin{enumerate}\itemsep=0pt
\item[$(i)$] The mechanical coordinates $S_{0i}$, $T_{0i}$, $S_{ij}$, $T_{ij}$ are parametrized by the separated coordinates $x_i$, $p_i$ and the values of the Casimir functions $C_i$, $K_i$, $i= 1,2$ as follows
\begin{subequations}\label{recons2}
\begin{gather}
S_{01}= -\frac{C_2}{4(D_1D_2)^{\frac{1}{2}}} \bigl((x_1-1) (x_1+1)^2 \big(4 C_1 x_2+C_2^2+C_2^2 x_2^2\big) p_1\nonumber\\
\hphantom{S_{01}=}{} -(x_2-1) (x_2+1)^2 \big(C_2^2 x_1^2+C_2^2+4 C_1 x_1\big) p_2\bigr),
\\
S_{11}= -\frac{C_2}{2 (D_1D_2)^{\frac{1}{2}}} \big((x_1+1) (x_1-1)^2 \big(4 C_1 x_2+C_2^2+C_2^2 x_2^2\big) p_1\nonumber\\
\hphantom{S_{11}=}{} -(x_2+1) \big(x_2-1)^2 (C_2^2 x_1^2+C_2^2+4 C_1 x_1\big) p_2\big),
\\
S_{21}= -\frac{i K_2 c_{11}}{8 (x_1 x_2)^{\frac{1}{2}}(D_1D_2)^{\frac{1}{2}}} (x_1-x_2) \big(4 C_1 x_2+C_2^2+C_2^2 x_2^2\big) \nonumber\\
\hphantom{S_{21}=}{}\times \big(C_2^2 x_1^2+C_2^2+4 C_1 x_1\big),
\\
S_{12}= - \frac{2 i(x_1 x_2)^{\frac{1}{2}}}{(x_1-x_2) c_{11} K_2 (D_1D_2)^{\frac{1}{2}}}
\bigl((x_1^2-1)^2 \big(4 C_1 x_2+C_2^2+C_2^2 x_2^2\big) p_1^2\nonumber\\
\hphantom{S_{12}=}{}- 2 \big(x^2_2-1\big) (x^2_1-1) \big(C_2^2 x_1 x_2+2 C_1 x_1+2 C_1 x_2+C_2^2\big) p_2 p_1 \nonumber\\
\hphantom{S_{12}=}{} +\big(x^2_2-1\big)^2\big(C_2^2 x_1^2+C_2^2+4 C_1 x_1\big) p_2^2\bigr),
\\
T_{01}= \frac{ i (x_2-1) (x_1-1) K_2}{4 (x_1 x_2)^{\frac{1}{2}} (D_1D_2)^{\frac{1}{2}}} \bigl((x_1+1) x_1 \big(4 C_1 x_2+C_2^2+C_2^2 x_2^2\big) p_1\nonumber\\
\hphantom{T_{01}=}{} -(x_2+1) x_2 \big(C_2^2 x_1^2+C_2^2+4 C_1 x_1\big) p_2\big),
\\
T_{11}= \frac{i (x_1+1) (x_2+1) K_2}{2 (x_1 x_2)^{\frac{1}{2}} (D_1D_2)^{\frac{1}{2}}}\bigl(x_1 (x_1-1) \big(4 C_1 x_2+C_2^2+C_2^2 x_2^2\big) p_1\nonumber\\
\hphantom{T_{11}=}{} -x_2 (x_2-1) \big(C_2^2 x_1^2+C_2^2+4 C_1 x_1\big) p_2\bigr),
\\
T_{21}=
 -\frac{c_{11}}{8 (x_1-x_2) \big(x^2_1-1\big) \big(x^2_2-1\big) p_2 p_1 C_2 x_1 x_2 (D_1D_2)^{\frac{1}{2}}}\nonumber\\
 \hphantom{T_{21}=}{} \times\bigl(x_1^2 \big(x^2_1-1\big)^2 \big(4 C_1 x_2+C_2^2+C_2^2 x_2^2\big)^2 \big({-}K_2^2 x_2^2+4 K_1 x_2-K_2^2\big) p_1^2\nonumber\\
\hphantom{T_{21}=}{} -2 x_1 x_2 \big(x^2_2-1\big) \big(x^2_1-1\big)
 \big(4 C_1 x_2+C_2^2+C_2^2 x_2^2\big) \big(C_2^2 x_1^2+C_2^2+4 C_1 x_1\big) \nonumber\\
\hphantom{T_{21}=}{} \times \big({-}x_1 K_2^2 x_2+2 x_1 K_1+2 K_1 x_2-K_2^2\big) p_2 p_1\nonumber\\
\hphantom{T_{21}=}{}+x_2^2 \big(x^2_2-1\big)^2 \big(C_2^2 x_1^2+C_2^2+4 C_1 x_1\big)^2 \big({-}K_2^2 x_1^2+4 x_1 K_1-K_2^2\big) p_2^2\bigr),
\\
T_{12}= \frac{2 (x_1-x_2) \big(x^2_1-1\big) \big(x^2_2-1\big) p_2 p_1 C_2}{c_{11} (D_1D_2)^{\frac{1}{2}}},
\\
S_{01}=S_{02},\qquad S_{22}=-S_{11}, \qquad T_{02}=T_{01}, \qquad T_{22}=-T_{11},
\end{gather}
\end{subequations}
where
\begin{gather*}
D_1=\big(x^2_1-1\big) \big(4 C_1 x_2+C_2^2+C_2^2 x_2^2\big) p_1-\big(x^2_2-1\big) \big(C_2^2 x_1^2+C_2^2+4 C_1 x_1\big) p_2,\\
D_2=x_1 \big(x^2_1-1\big) \big(4 C_1 x_2+C_2^2+C_2^2 x_2^2\big) p_1-
x_2 \big(x^2_2-1\big) \big(C_2^2 x_1^2+C_2^2+4 C_1 x_1\big) p_2.
\end{gather*}

\item[$(ii)$] If the Poisson relations among $p_i$ and $x_j$ have quasi-canonical form~\eqref{quascan2'} then the variables $S_{0i}$, $T_{0i}$, $S_{ij}$, $T_{ij}$ given by~\eqref{recons2} satisfy the Poisson brackets~\eqref{trigSklST}.
\end{enumerate}
\end{Theorem}

\begin{proof}[Idea of the proof] Item (i) of the theorem is proven by the direct calculations. The item (ii) follows from our general theory and can be also checked by the direct verifications.
\end{proof}

\subsubsection{The Abel-type equations}\label{section5.4.4}
In the case $c_{22}=0$, $c_{11}\neq 0$ the Abel equations written for the separated coordinates $x_i$ are trivial due to the fact that $x_i$ in this case are expressed via the integrals of motion. The non-trivial Abel-type equations are written for the conjugated momenta~$p_i$
\begin{gather*}
 \dfrac{1}{2 x_1 I_0+I_1}
\frac{1}{p_1}\frac{{\rm d}p_1}{{\rm d}t_1} + \dfrac{1}{2 x_2 I_0+I_1}
\frac{1}{p_2}\frac{{\rm d}p_2}{{\rm d}t_1} =1,
\\
 \dfrac{1+x_1^{2} \dfrac{ c_{11}^2 c_1c_2}{I^2_2}}{ 2 x_1 I_0+I_1}
\frac{1}{x_1p_1}\frac{{\rm d}p_1}{{\rm d}t_1} + \dfrac{1+x_2^{2} \dfrac{ c_{11}^2 c_1c_2}{I^2_2}}{ 2 x_2 I_0+I_1}
\frac{1}{x_2p_2}\frac{{\rm d}p_2}{{\rm d}t_1} =0,
\\
 \dfrac{1}{2 x_1 I_0+I_1}
\frac{1}{p_1}\frac{{\rm d}p_1}{{\rm d}t_2} + \dfrac{1}{2 x_2 I_0+I_1}
\frac{1}{p_2}\frac{{\rm d}p_2}{{\rm d}t_2} =0,
\\
 \dfrac{1+x_1^{2} \dfrac{ c_{11}^2 c_1c_2}{I^2_2}}{ 2 x_1 I_0+I_1}
\frac{1}{x_1p_1}\frac{{\rm d}p_1}{{\rm d}t_2} + \dfrac{1+x_2^{2} \dfrac{ c_{11}^2 c_1c_2}{I^2_2}}{ 2 x_2 I_0+I_1}
\frac{1}{x_2p_2}\frac{{\rm d}p_2}{{\rm d}t_2} =1.
\end{gather*}

Introducing the angle variables
$p_i=\exp{\phi_i}, i= 1,2$ we obtain that
\begin{gather*}
 \phi_1 = \phi_1^0-\frac{i c_{11} K_2 C_2 (x_1-x_2)}{2 (x_1x_2)^{\frac{1}{2}}} \big(t_1-t^0_1\big)+ \frac{1}{2}ic_{11} K_2C_2 (x_1x_2)^{\frac{1}{2}} \big(t_2-t^0_2\big) ,
\\
 \phi_2 = \phi_2^0-\frac{i c_{11} K_2 C_2 (x_2-x_1)}{2 (x_1x_2)^{\frac{1}{2}}} \big(t_1-t^0_1\big)+ \frac{1}{2}ic_{11} K_2C_2 (x_1x_2)^{\frac{1}{2}} \big(t_2-t^0_2\big) ,
\end{gather*}
where $x_i$, $i= 1,2$ are the constants of motion calculated from the equations~(\ref{degsepeq2'}) and we have taken into account that there exists the following relations among the integrals $I_0$, $I_1$, $I_2$, Casimir functions $C_2$, $K_2$ and the action coordinates $x_1$, $x_2$
\begin{gather*}
I_1 = \frac{ic_{11}(x_1+x_2)K_2C_2}{(x_1x_2)^{\frac{1}{2}}}, \qquad I_0 = -\frac{ic_{11} K_2C_2}{(x_1x_2)^{\frac{1}{2}} }, \qquad I_2=-ic_{11}K_2C_2 (x_1x_2)^{\frac{1}{2}}.
\end{gather*}

\section{Conclusion and discussion}\label{section6}
In the present paper we have constructed non-standard separated variables for the classical integrable XXX and XXZ spin chains with the degenerated twist matrix reversing the role of the functions $A(u)$ and $B(u)$ in the corresponding separating algebra. We have also obtained the Abel-type equations for the corresponding coordinates/momenta of separation. We have shown that for special cases of the degenerated twist matrices our separated variables can be identified with the action-angle coordinates from the Liouville theorem.

It would be interesting and physically important to find quantum and higher rank analogues of the proposed separated variables. These problems are open.

\pdfbookmark[1]{References}{ref}
\LastPageEnding

\end{document}